\documentclass[journal,twoside,web]{ieeecolor}

\usepackage{generic}
\usepackage{cite}
\usepackage{amsmath,amssymb,amsfonts}
\usepackage{graphicx}
\usepackage{algorithm}
\usepackage{hyperref}
\hypersetup{hidelinks=true}
\usepackage{amsmath,amssymb}

\usepackage{amsmath}

\usepackage{bm}
\usepackage{psfrag}
\usepackage{epsfig}
\usepackage{graphicx,subfigure}
\usepackage{mathrsfs}
\usepackage{url}
\usepackage{mathrsfs}
\usepackage{amsthm}
\usepackage{bbm,color}
\usepackage{pifont}
\usepackage{multirow}
\usepackage{supertabular}
\usepackage{booktabs}
\usepackage{overpic}
\usepackage{algorithm}
\usepackage{algpseudocode}
\usepackage{amsmath}
\usepackage{graphics}
\usepackage{graphicx}
\usepackage{epsfig}
  \usepackage{stfloats}
 
\usepackage{stmaryrd}
\usepackage{cuted}
\usepackage{textcomp}
\def\RR{{\mathbb R}}    

\newcommand{\overbar}[1]{\mkern 1.5mu\overline{\mkern-1.5mu#1\mkern-1.5mu}\mkern 1.5mu}

\def\bx{\textbf{x}}

\def\cK{\mathcal{K}}
\def\cL{\mathcal{L}}
\def\cKL{\mathcal{KL}}

\newtheorem{definition}{\textbf{Definition}}
 \newtheorem{assumption}{\textbf{Assumption}}
 \newtheorem{proposition}{\textbf{Proposition}}
  \newtheorem{corollary}{\textbf{Corollary}}
 \newtheorem{lemma}{\textbf{Lemma}}
 \newtheorem{theorem}{\textbf{Theorem}}
 \newtheorem{remark}{\textbf{Remark}}
  \newtheorem{example}{\textbf{Example}}

\def\BibTeX{{\rm B\kern-.05em{\sc i\kern-.025em b}\kern-.08em
    T\kern-.1667em\lower.7ex\hbox{E}\kern-.125emX}}
\begin{document}

\title{A Nonlinear Incremental Approach for Replay Attack Detection}
\author{Tao Chen, Andreu Cecilia, Lei Wang*, Daniele Astolfi, Zhitao Liu
\thanks{Tao Chen, Lei Wang and Zhitao Liu are with the College of Control Science and Engineering, Zhejiang University, P.R. China (e-mail: tao\_chen; lei.wangzju; ztliu@zju.edu.cn). Andreu Cecilia is with the Universitat Polit\'ecnica de Catalunya, Avinguda Diagonal, 647, 08028 Barcelona, Spain. (e-mail: andreu.cecilia@upc.edu).Daniele Astolfi is with the Univ. Lyon, Université Claude Bernard
Lyon 1, CNRS, LAGEPP UMR 5007, F-69100 Villeurbanne, France (e-mail: daniele.astolfi@univ-lyon1.fr). Corresponding author: Lei Wang}
}

\maketitle

\begin{abstract}
Replay attacks comprise  replaying previously recorded sensor measurements and injecting malicious signals into a physical plant, causing great damage to cyber-physical systems. Replay attack detection has been widely studied for linear systems, whereas  limited research has been reported for nonlinear cases. In this paper, the replay attack is studied in the context of a nonlinear plant controlled by an observer-based output feedback controller. We first analyze replay attack detection using an innovation-based detector and reveal that this detector alone may fail to detect such attacks. Consequently, we turn to a watermark-based design framework to improve the detection. In the proposed framework, the effects of the watermark on  attack detection and closed-loop system performance loss are  quantified by two indices, which exploit the incremental gains of nonlinear systems. To balance the detection performance and control system performance loss, an explicit  optimization problem is formulated. 
Moreover, to achieve a better balance, we generalize the proposed watermark design framework   to  co-design the watermark, controller and observer.  Numerical simulations are presented to validate the proposed frameworks. 
\end{abstract}

\begin{IEEEkeywords}
Cyber-physical systems, nonlinear systems, replay attack detection, incremental gains. 
\end{IEEEkeywords}

\section{Introduction}\label{sec.Introduction}
 
Cyber-physical systems (CPSs) seamlessly integrate computational algorithms with physical components \cite{kim2022survey}.
Such integrations rely on (wireless) communication networks, raising   security risks.
To analyze and improve the resilience of CPSs to security risks, multiple types of attacks are studied in the literature, e.g., false data injection attack \cite{habib2023false,zhang2020false}, replay attack \cite{mo2009secure,langner2011stuxnet}, deny-of-service attack \cite{lu2017input,li2020active}, eavesdropping attack \cite{wang2018jamming}, etc. Among them, the replay attack has received significant attention due to its simplicity, stealthiness and model-free nature.

For attack detection, one commonly used detector is the $\chi^2$ detector \cite{MEHRA1971637,1997A}, which is innovation-based, that is, it relies on the discrepancy between the predicted and received sensor outputs.  However, as demonstrated in  \cite{mo2009secure,mo2013detecting},  an innovation-based detector may fail to detect replay attacks in the sense that, under some stability conditions, its detection rate converges to the false alarm rate. To address this issue, lots of  mechanisms have been proposed in the literature, which can be broadly classified  into two categories. The first involves ``encoding-decoding'' the sensor measurement to be replayed. Specifically, these methods use a random pair (e.g., two identical random numbers) to encode the sensor measurements to be transmitted and then decode it at the receiver's side \cite{du2022attack,li2023dynamic,joo2020resilient}. Under a replay attack, the detection is facilitated by the mismatch of the random pair. However, to achieve synchronization of the random pair,  prior information available to both sides \cite{du2022attack,li2023dynamic} or an extra  communication channel \cite{joo2020resilient}  is required, which  reduces the applicability
and the inherent robustness against adversaries \cite{ferrari2020switching}.

 An alternative to detect a replay attack is to add a ``watermark'' \cite{mo2009secure} to the control input and stimulating the physical plant.   Under a replay attack, the incoherence between the expected  stimulation and the actual sensor measurements facilitates the replay attack detection. Nonetheless, since the watermark is essentially a disturbance directly acting on the physical plant, it inevitably affects the control system performance. Therefore, both detection performance and  control system performance loss are considered when designing a watermark.
 
In \cite{mo2009secure}, the detection performance is evaluated by the difference between the healthy innovation covariance and the attacked one, while the system performance loss is evaluated by the watermark-induced extra linear-quadratic-Gaussian (LQG) cost. Then, based on these evaluations,  an optimization problem is constructed in \cite{mo2013detecting} to balance these two factors. Motivated by these results, various strategies have been reported  to further reduce the control system performance loss or, equivalently, to enhance detection performance with the same amount of watermarking. For example,  \cite{FANG2020108698} proposes a periodic watermarking scheduling approach. In \cite{naha2023quickest}, a parsimonious policy is proposed to limit the average number of watermarking events. In \cite{zhao2023stochastic}, an event-based physical watermark is designed, where the probability of adding the watermark is determined by the innovation.  In \cite{chen2023replay}, sensitive states are considered, and  additional constraints are developed to strictly limit the effect of the watermark on these states. 
Moreover, numerous watermark design methods have been developed to achieve specific objectives in replay attack detection, e.g.,   minimizing the average detection delay  \cite{naha2023quickest},    making the watermark unpredictable \cite{zhai2021switching}, extending the watermark for  industrial process operation optimization of cyber-physical systems \cite{yang2022joint}, and  automatically learning and recognizing replay attacks \cite{yu2023reinforcement}. 

To sum up, watermark approaches have been widely used to enhance $\chi^2$ detectors. Moreover, the detection performance and the control system performance loss have been well quantified and balanced. 
However, all  of these studies have been conducted in the context of linear systems, and it is not straightforward to extend these methods to nonlinear cases. For example, in \cite{mo2009secure,zhai2021switching,FANG2020108698,zhao2023stochastic,chen2023replay}, the detection performance is evaluated by the difference between the attacked and the healthy innovation covariance, which can be explicitly computed for linear cases, but is significantly more complex in nonlinear scenarios.

In view of the above discussion, in this work, we  study the replay attack detection problem for nonlinear systems. Specifically, we first analyze the detection performance with an innovation-based detector and show that, by itself,  it may fail to detect such attacks. We therefore  adopt a watermark-based design framework. {To evaluate the detection performance and control system performance loss,   we employ incremental gain-based methods, which are   broadly applicable in nonlinear scenarios. Furthermore,   incremental gains   allow us to bound the output (or state) of a system  without requiring explicit calculation of variable distributions and independently of the system's  equilibrium point.}   With these   evaluations, sufficient conditions for specific detection performance and control system performance loss can be rigorously established and exploited, which support the   watermark design. The main contributions are summarized as follows.
\begin{itemize}
    \item  A new watermark design framework for  nonlinear systems replay attack detection   is proposed, where the detection performance and the control system performance loss caused by the watermark are evaluated by some indices that are induced by incremental gains.  
    \item Sufficient conditions for specific detection performance and control system performance loss are established, thereby enabling the construction of a  solvable optimization problem for the watermark design. 
    \item { A systematic approach { based on LMIs} for co-designing the watermark, the controller and the observer is further developed to achieve a better tradeoff between the detection performance and the control performance loss.}

    \item For detection performance evaluation, a lower bound between the input difference and the output   difference for   nonlinear systems is required. However, this index, referred to as the incremental $\mathcal{L}_2^-$ gain and denoted as $\mathcal{L}_{\delta 2}^-$, has not been well developed in the literature.  Extending the $\mathcal{L}_2^-$ gain to $\mathcal{L}_{\delta 2}^-$ and  developing its corresponding  Lyapunov  (Proposition 2)   and LMI (Lemma 2)  characterizations are two  contributions of the work.
    \item     The connection between detection performance and $\mathcal{L}^-_{\delta 2}$ gain is not straightforward. Establishing a clear link between the   detector and the $\mathcal{L}^-_{\delta 2}$ gain for effective performance evaluation is another contribution. 
\end{itemize}

\vspace{1em} 
 \noindent{\bf Notation}. We denote by $\mathbb{N}$ the set of natural numbers, $\mathbb{R}^n$  the set of real numbers of dimension $n\in\mathbb{N}$, $\mathbb{R}_{\geq 0}$  the set of non-negative reals, $\mathbb{R}_{> 0}$ the set of positive reals.
    For a vector $x\in\mathbb{R}^n$, $\|x\|$ denotes the { Euclidean} norm and 
   $\|x\|_\Lambda:=\sqrt{x^\top \Lambda x}$ for $\Lambda \in\mathbb{R}^{n\times  n} $. For random vector $x$,
$x\sim \mathcal{N}(\mu,\Sigma)$ denotes that $x$ follows a Gaussian distribution with mean $\mu$ and covariance $\Sigma$.
    For column vectors $x\in \mathbb{R}^m$ and $y\in \mathbb{R}^n$, $\text{col}(x,y): = [x^\top, y^\top]^\top$. 
  For matrix $A\in\mathbb{R}^{n\times n}$, $\lambda_\text{max}(A)$ denotes the largest eigenvalue of $A$ ($\lambda_\text{min}(\cdot)$ for the minimum). For matrices $A\in\mathbb{R}^{m\times m}$ and $B\in\mathbb{R}^{n\times n}$, $\text{diag}(A,B)$ denotes a matrix with  $A$ and $B$ in the main diagonal and $0$ everywhere else.  $I_{n}$ denotes identity matrix of dimension $n \times n$. For real-valued Lebesgue
integrable functions  $f:\mathbb{R}_{\geq t_0}\rightarrow \mathbb{R}^n$, ${\|f\|_{\mathcal{L}_2}:=\int_{t_0}^\infty \|f(t)\|^2 \text{d}t }$ for some $t_0\in\RR_{\geq0}$.  
Given a function $u:\RR_{\geq0}\to \RR^m$, we define
$\|u(\cdot)\|_\infty = \text{sup}_{t\in[0,\infty)}\|u(t)\|$.
Given a  signal $u:\RR_{[0,\infty)}\to\RR^m$
and a scalar $\tau\in [0,\infty)$,
we denote by 
$(u(t))_\tau$ the $\tau$-truncation of $u(t)$ defined as  $ (u(t))_\tau=u(t)$
for all $t\in [0,\tau]$ and $(u(t))_\tau=0$,
for all $t> \tau$. \color{black}
{A mapping $f : \mathbb{R}^p\rightarrow \mathbb{R}^q $ is  $\mathcal{C}^n$ if it is $n$-times continuously
differentiable.}
In addition, a continuous function $ \alpha : \mathbb{R}_{\geq 0}\rightarrow \mathbb{R}_{\geq 0}  $ is of class $\mathcal{K} $, if it is   strictly increasing and $ \alpha(0)=0$. A continuous function $\beta : \mathbb{R}_{\geq 0} \times \mathbb{R}_{\geq 0} \rightarrow \mathbb{R}_{\geq 0} $ is of class $\mathcal{K}\mathcal{L}$ if, for each fixed $s\geq 0$, the function $\beta(\cdot,s)$ is of class $\mathcal{K}$ and, for each fixed $r>0$,  $\beta(r,s)$ is strictly decreasing  and $\lim_{r\to \infty}\beta(r,s)\rightarrow 0$.


\section{Preliminaries on incremental gain properties}\label{sec.incremental_gain}

Consider a system of the form 
\begin{equation}
    \label{eq.general_sys}
    \dot x = f(x)+Bu, \qquad y = Cx+Du
\end{equation}

where $f:\RR^n\to \RR^n$
is a { (locally)} Lipschitz function satisfying $f(0) = 0$, 
$x\in \RR^n$ is the state, 
 $u:[0,\infty)\to \RR^m$ is a measurable and locally essentially
bounded function taking values on 
a set of $\mathcal{U}\subset\RR^m$
containing the origin, and $y\in \RR^p$ the output.
\color{black}
Furthermore, 
we denote by 
$X(t,x_0,u)$ the unique solution of system 
\eqref{eq.general_sys} at time $t$, with initial state $x_0\in \RR^n$ and input ${u(t)}\in\mathcal{U}$.
Similarly, we denote by $Y(t,x_0,u)= CX(t,x_0,u)+D u{(t)}$ its output trajectory.

For system \eqref{eq.general_sys}, we provide the    definitions of input-to-state stable {\cite{sontag1989smooth}}, denoted as ISS, and 
incremental input-to-state stable {\cite{angeli2002lyapunov}}, denoted as 
$\delta$ISS.

\begin{definition}[ISS]\label{def.iss}
System \eqref{eq.general_sys} is said to be input-to-state stable (ISS) if there exist  $\alpha\in \cK$ and $\beta\in \cKL$    such that the following holds {for all $t\geq0$}
\begin{equation}\label{eq.ISS}
    \|X(t,{x_0},u)\|
    \leq \beta(\|{x_0}\|,t) + \alpha({\|u(\cdot)  \|_\infty})
\end{equation}
for any 
${x_0} \in \RR^n$ and   $u(t)\in \mathcal U$.
\end{definition}

\begin{definition}[$\delta$ISS]\label{def.delta_iss}
System \eqref{eq.general_sys} is said to be incrementally input-to-state stable if there exist  $\alpha\in \cK$ and $\beta\in \cKL$    such that the following holds { for all $t\geq0$}
\begin{multline*}
    \|X(t,{x_1} ,u_1)- X(t,{x_2},u_2)\|
    \leq \beta(\|{x_1-x_2}|,t)
    \\ + \alpha({\|u_1(\cdot)-u_2(\cdot)\|_\infty})
\end{multline*}
for any 
${x_{1} , x_{2}} \in \RR^n$ and 

 $u_1(t), u_2(t)\in \mathcal U$.
\end{definition}


Next, { based on \cite{van2000l2,verhoek2023convex}}, we define the  incremental 
$\cL_2^+$ gain, denoted as 
 $ \cL_{\delta 2}^+$.

\begin{definition}[$\mathcal{L}^+_{\delta 2}$ gain]\label{def.L_ifty}
The  $\mathcal{L}^+_{\delta 2}$ gain of system \eqref{eq.general_sys} is defined as $\mathcal{L}^+_{\delta 2}:=\inf{\gamma^+}$,  
if   there exist    $\gamma^+<\infty$ and $\alpha^+ \in \mathcal{K} $ such 
that 
for all $t\geq0$ and   $\tau\geq0$,
\begin{multline}\label{eq.def_L_infty_gain}
    \|((Y(t,x_1,u_1)- Y(t,x_2,u_2))_{\tau}\|_{\cL_2}
    \leq 
    \\
    \gamma^+\|(u_1(t)-u_2(t))_\tau\|_{\mathcal{L}_2}
    +\alpha^+(\|x_1-x_2\|) 
\end{multline}
holds for any $x_1,x_2\in \RR^n$ and $u_1(t),u_2(t)\in \mathcal U$. 
\end{definition}

In the following, we also use $[\mathcal{L}^+_{\delta 2}]_y^u$ to indicate the input and output with respect to which the  $\mathcal{L}^+_{\delta 2}$ gain is defined.

 With the definition of $\mathcal{L}^+_{\delta 2}$, an upper bound on the $\cL_2^+$ norm of the error between two different output trajectories of system \eqref{eq.general_sys} can be effectively estimated.
Correspondingly, to get a lower bound, we turn to the $\mathcal{L}^-_{\delta 2}$ gain, which is an incremental version of the $\mathcal{L}_2^-$ gain \cite{liu2005lmi}. The  formal definition of the $\mathcal{L}^-_{\delta 2}$ gain is provided below.

\begin{definition}[$\mathcal{L}^-_{\delta 2}$ gain]\label{def.L_minus}
The $\mathcal{L}^-_{\delta 2}$ gain of
system \eqref{eq.general_sys}
is defined as $\mathcal{L}^-_{\delta 2}:=\sup \gamma^-$,   if   there exist    $\gamma^->0$ and $\alpha^- \in \mathcal{K} $ such 
that 
for all $t\geq0$ and  $\tau\geq0$,
\begin{multline}\label{eq.def_L_minus_gain}
    \|((Y(t,x_1,u_1)- Y(t,x_2,u_2))_{\tau}\|_{\cL_2}
    \geq 
    \\
    \gamma^-\|(u_1(t)-u_2(t))_\tau\|_{\mathcal{L}_2}
    -\alpha^-(\|x_1-x_2\|) 
\end{multline}
holds for any  $x_1,x_2\in \RR^n$ and  $u_1(t),u_2(t)\in \mathcal U$. 
\end{definition}

\begin{remark}

Note that for system \eqref{eq.general_sys}
$\delta$ISS implies an $\mathcal{L}_{\delta 2}^+$ gain because the output is linear (this may be not hold for generic nonlinear functions).
However, the converse is true only under
additional (differential) detectability conditions.
%
\end{remark}

\begin{remark}

Incremental properties, such as those considered in this section, relate  the state and output differences between any pair of the system's trajectories with distinct  initial states and inputs. 

As extensively studied in the literature
(e.g. \cite{angeli2002lyapunov,romanchuk1996characterization}) 
one has the implications 
$\delta$ISS $\Rightarrow$ ISS
and $\mathcal{L}_{\delta 2}^+$ gain $\Rightarrow$ $\mathcal{L}_{2}$ gain, 
while the converse generically does not hold true.
\end{remark}
 
In view of the previous definitions, we  now provide  two Lyapunov characterizations of the $\cL_{\delta 2}^+$ and $\cL_{\delta 2}^-$ gains. { The proof of 
the former can be found in \cite[Appendix B.1]{verhoek2023convex}, while the proof of the latter is postponed to Appendix \ref{prof.def.Lyapunov_L_minus}.}
Concerning the Lyapunov characterization of $\delta$ISS we refer to \cite{angeli2002lyapunov}.

\begin{proposition}[Lyapunov $\mathcal{L}^+_{\delta 2}$  characterization]\label{def.Lyapunov_L_ifty}
    Suppose there exist a  { $\mathcal{C}^1$} function 
    $V^+:\RR^n\times \RR^n\times\RR\to \RR_{\geq0}$ ,
     functions 
    $\underline\alpha^+, \overbar \alpha^+\in \cK_\infty$
    and $\gamma^+>0$ such that 
    \begin{equation}\label{eq.sandwich_Vplus}
\underline\alpha^+(\|x_1-x_2\|) \leq V^+( x_1,x_2,t)\leq \overbar \alpha^+(\|x_1-x_2\|)
    \end{equation}
      \begin{multline}\label{eq.ineq_Vplus}
\dfrac{\partial V^+}{\partial t}
(x_1,x_2,t)
+
\dfrac{\partial V^+}{\partial x_1}(f(x_1 )+Bu_1)
+
\dfrac{\partial V^+}{\partial x_2}(f(x_2)+B  u_2)
\\
\leq \gamma^+\|u_1-u_2\|^2 - \|y_1-y_2\|^2
    \end{multline}
        for all 
    $t\geq0$,
    $x_1,x_2\in \RR^n$, 
    and $u_1,u_2\in \mathcal{U}$.
     Then the $\mathcal{L}^+_{\delta 2}$ gain of  system \eqref{eq.general_sys} satisfies 
  $[\mathcal{L}^+_{\delta 2}]^u_y\leq  \gamma^+$.
\end{proposition}

\begin{proposition}[Lyapunov $\mathcal{L}^-_{\delta 2}$  characterization]\label{def.Lyapunov_L_minus}
    Suppose there exist a { $\mathcal{C}^1$}  function 
    $V^-:\RR^n\times \RR^n\times\RR\to \RR_{\geq0}$,
     functions 
    $\underline\alpha^-, \overbar \alpha^-\in \cK_\infty$ and 
    $\gamma^->0$ such that 
    \begin{equation}\label{eq.sandwich_Vminus}
\underline\alpha^-(\|x_1-x_2\|) \leq V^-(  x_1,x_2, t)\leq \overbar \alpha^-(\|x_1-x_2\|)
    \end{equation}
      \begin{multline}\label{eq.ineq_Vminus}
\dfrac{\partial V^-}{\partial t}
(x_1,x_2,t)
+
\dfrac{\partial V^-}{\partial x_1}(f(x_1)+B u_1)
+
\dfrac{\partial V^-}{\partial x_2}(f(x_2 )+Bu_2)
\\
\leq -\gamma^-\|u_1-u_2\|^2 + \|y_1-y_2\|^2
    \end{multline}
    for all 
    $t\geq0$,
    $x_1,x_2\in \RR^n$, 
    and $u_1,u_2\in \mathcal{U}$.
      Then the  $\mathcal{L}^-_{\delta 2}$ gain of system \eqref{eq.general_sys} 
   satisfies
    $ [\mathcal{L}^-_{\delta 2}]^u_y\geq  \gamma^-$.
\end{proposition}

Finally, for systems of form \eqref{eq.general_sys}, {if we assume that {$f: \RR^n\rightarrow \RR^n$ is  $\mathcal{C}^1$}},
we can provide matrix inequality 
characterizations of the previous incremental gains, which will be useful in the subsequent development of a computationally viable design methodology.

{A matrix inequality 
characterization for the $[\mathcal{L}^+_{\delta 2}]^u_y$ gain follows from \cite[Corollary 14]{verhoek2023convex}, and is recalled in the next Lemma.} For convenience, define 
      $A_x := \frac{\partial f}{\partial x}(x)$.
  
\begin{lemma}\label{lem.L_inf_ori}
Suppose  there exist a symmetric positive definite matrix $P\in\RR^{n\times n}$ 
    and $\gamma^+\geq0$
    satisfying
\begin{equation}\label{eq.LMI_L_inf_ori}
        \begin{bmatrix}
       A_x^\top P\!+\!PA_x \!+\!C^\top C&PB   \!+\!C^\top  D   \\
        \star &D^\top D\!-\!\gamma^+\! I_m
    \end{bmatrix}\! \preceq \! 0,
    \end{equation}
   for all  $x\in \RR^n $.
   Then,
   $V^+(x_1,x_2) = (x_1-x_2)^\top P (x_1-x_2)$
   satisfies \eqref{eq.sandwich_Vplus}, \eqref{eq.ineq_Vplus}
   and 
   system \eqref{eq.general_sys} 
    has an incremental 
    $\mathcal{L}^+_{\delta 2}$ gain $ [\mathcal{L}^+_{\delta 2}]^u_y\leq  \gamma^+$.
  \end{lemma}

{Motivated by the previous result, we propose a similar characterization for the  $ [\mathcal{L}^-_{\delta 2}]^u_y$ gain in the next Lemma.}

\color{black}
\begin{lemma}\label{lem.L_minus_ori}
 Suppose there exist a symmetric negative definite matrix $Q\in\RR^{n\times n}$
      and $\gamma^-\geq0$
      satisfying   
\begin{equation}\label{eq.LMI_L_minus_ori}
        \begin{bmatrix}
      A_x^\top Q\!+\!QA_x\!+\!C^\top C&QB \!+\!C^\top D\\
         \star  &D^\top D\!-\!\gamma^-I_m 
    \end{bmatrix}\succeq 0 
    \end{equation}
    for all  $x\in \RR^n$.
    Then,     $V^-(x_1,x_2) = -(x_1-x_2)^\top Q (x_1-x_2)$
   satisfies \eqref{eq.sandwich_Vminus}, \eqref{eq.ineq_Vminus}
   and system \eqref{eq.general_sys} 
    has an incremental 
    $\mathcal{L}^-_{\delta 2}$ gain $ [\mathcal{L}^-_{\delta 2}]^u_y\!\geq\!  \gamma^-$.
\end{lemma}

The proof of Lemma~\ref{lem.L_minus_ori} is postponed to Appendix \ref{prof.Hi-}.

 \begin{remark}\itshape
    From \eqref{eq.LMI_L_minus_ori}, it is evident that if $D  =0$,   then $\gamma^-$ is always $0$. To obtain a positive $\gamma^-$, the relative degree of system \eqref{eq.general_sys} between $u$ and $y$ must be $0$. For systems whose relative degree is greater than $0$, one can reduce its relative degree to $0$ by 
    redefining a new output
    \color{black} to
    include an   auxiliary direct channel \cite{liu2005lmi}.
\end{remark} 

\section{Problem Formulation}\label{sec.Problem_formu} 
\subsection{System Description}

Consider  continuous-time nonlinear plants of the form
 \begin{equation}\label{eq.plant}
\Sigma_{np}:
\; \left\{
\begin{aligned}
  \dot{x}&= f(x) +Bu+\omega\\
  y  &= Cx+Du +\nu ,
\end{aligned}
\right.
 \end{equation}
where $x\in  \mathbb{R}^{n}$ is the plant state, $u \in   \mathbb{R}^{m}$ is the input signal,  $y\in  \mathbb{R}^{p}$ is the sensor measurement,  and $\omega\in \mathbb{R}^{n}$ and $\nu\in \mathbb{R}^{p}$ are the system and sensor measurement noise, respectively, which are assumed to be Lebesgue integrable and bounded as $\|\omega\|_\infty\leq  \overbar \omega $ and $\|\nu \|_\infty\leq  \overbar \nu $  with $\overbar \omega >0 $ and $\overbar \nu >0 $. Finally,   $f: \mathbb{R}^n  \rightarrow \mathbb{R}^n$ {is  $\mathcal{C}^1$}.

 Furthermore,  
we denote by 
$X(t,x_0,u,\omega)$ the unique solution of system 
\eqref{eq.plant} at time $t$, with initial states $x_0\in \RR^n$
and subject to the input $u$ and noise $\omega$.
Similarly, we denote by $Y(t,x_0,u, \omega,\nu) = CX(t,x_0,u,\omega)+Du(t)+\nu(t)$ its output trajectory. For convenience, we omit the noise signal $\omega$ and $\nu$ and use notations $X(t,x_0,u)$ and $Y(t,x_0,u)$ in the following.

For output feedback control purpose, we assume that a Luenberger-type observer \cite{mohan2017synthesizing, le2015privacy} and a state feedback controller are deployed, taking the form of 
\begin{equation}\label{eq.observer}
\Sigma_{o}:
\; \left\{\begin{aligned}
\dot{\hat x}&= f(\hat x)+Bu+ L(y-\hat y)\\
  \hat y  &=C\hat x+Du,
  \end{aligned}\right.
\end{equation} 
and 
  \begin{equation}\label{eq.controller}
\begin{split}
\Sigma_{c}: u=\kappa(\hat x) +v
 \end{split}
\end{equation} 
 respectively, where $\hat x \in \mathbb{R}^n$ and $\hat y\in    \mathbb{R}^p$ are the estimated state and output, respectively,  $L\in \mathbb{R}^{n\times p}$ is the gain of observer, $\kappa: \mathbb{R}^{n}\rightarrow \mathbb{R}^m$ is a Lipschitz feedback map, that is,  there exists a constant  $l_{\kappa}>0$   such that
\begin{equation}\label{eq.lip_k}
   \|\kappa(x_1) - \kappa(x_2)\|\leq l_{\kappa}\|x_1-x_2\| , 
\end{equation}
for all $  x_1,x_2 \in \RR^n$, and $v$   denotes a residual control signal to be addressed later.  We denote by 
$\widehat X(t,\hat x_0,v , y )$ and $\widehat Y(t,\hat x_0,v , y )=C\widehat X(t,\hat x_0,v,y)+Du $ the solution and the output trajectory of the observer \eqref{eq.observer}, respectively.

Define  $\tilde x:= x-\hat x$ as the    observation error,   whose dynamics  are given as follows
\begin{equation}
    \label{eq.error_observer}
    \begin{aligned}
           \dot {\tilde x} & =  f(x )- f(x-\tilde x  ) + L \tilde y+\omega,\\
           \tilde y & = C \tilde x   +\nu  
    \end{aligned}
\end{equation}
where  $ u  = \kappa(x-\tilde x)+v$. We denote by $\widetilde X (t,\bx_0 ,v,y )=X(t,x_0,u)-\widehat X(t,\hat x_0,v , y )$ the solution   and $\widetilde Y(t,    \bx_0,v,y)=Y(t,  x_0,u )-\widehat Y(t,\hat  x_0, v , y )$ the output trajectory of  \eqref{eq.error_observer}, respectively, where $\bx_0 := \text{col}(x_0,\hat  x_0)$.

{In addition, with the controller \eqref{eq.controller}, the physical plant \eqref{eq.plant} evolves as follows}
{\begin{equation}\label{eq.plant_controller}
\dot x = f(x) + B(\kappa(x-\tilde x) +v)+\omega.
\end{equation}}

Then, we assume that the error dynamic system \eqref{eq.error_observer} and the controlled  plant \eqref{eq.plant_controller} and satisfy  input-to-state stability (ISS) property \cite{sontag1997output,angeli2002lyapunov}. That is,
{\begin{assumption}\label{ass.delta_ISS}
    The error dynamic system \eqref{eq.error_observer} is  ISS w.r.t. inputs $\nu$ and $\omega$ (see Definition~\ref{def.iss}) { uniformly on $x$}, and the controlled plant \eqref{eq.plant_controller} is $\delta$ISS  w.r.t. inputs $\tilde x, v,  \omega$ (see Definition~\ref{def.delta_iss}).
\end{assumption}}

Designing controllers  and observers   to  achieve Assumption~\ref{ass.delta_ISS} has been well studied, see for instance \cite{giaccagli2023lmi}, and a methodology can be found in Lemma~\ref{lem.contraction_1}.



 {\begin{remark}
 We analyze the nonlinear system of the form \eqref{eq.plant} and assume a linear output, as both choices are necessary to develop computationally tractable conditions.
While parts of this theory might extend to more general (e.g., non-smooth) systems or those with nonlinear outputs, these generalizations would likely prevent the tractable LMI characterizations for $\mathcal{L}^-_{\delta 2}$ and $\mathcal{L}^+_{\delta 2}$ gains presented in Lemmas 1 and 2. For instance, a nonlinear output would likely require state-dependent matrices $P(x)$ and $Q(x)$, rendering the conditions intractable.
\end{remark}}

\subsection{Communication Topology and Replay Attack}

In practice, for remote control purpose, the controller and observer may be deployed in a remote control center. In this case, as shown in Fig. \ref{fig.replay_fram}, the sensor measurement $y$ and the control signal $u$ are transmitted through networks, making them susceptible to attacks.

\begin{figure}[h]
\centering
\includegraphics[width=1\columnwidth]{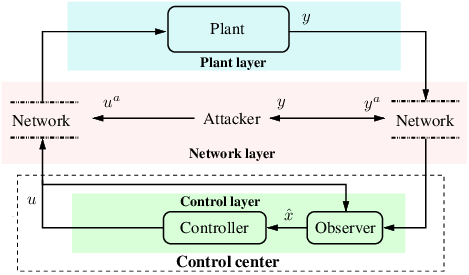}
\caption{Scheme of the communication topology.}
\label{fig.replay_fram}
\end{figure}

In this paper, we suppose that the network suffers from a \emph{replay attack}, and is free from delays and noise for simplicity.
The following resources are supposed to be available to the adversary:
\begin{itemize}
\item[(i)] The adversary can monitor and record the sensor measurement $y $ for all time $t$.
\item[(ii)] The adversary can arbitrarily modify the transmitted signals $u $ and $y $  to $u^a$ and $ y^a  $, respectively.
\end{itemize}

With the above resources, the replay attack strategy is given as follows. 
\begin{enumerate}
  \item \textbf{(Record)} From time $0$ to time $T$, the adversary records a sequence of sensor measurements $y$.
  \item \textbf{(Replay)} From time $T $ to time $2T$, the adversary replays the recorded sensor measurements to tamper the true sensor outputs, i.e.,
\begin{equation}\label{eq.replay_attack}
  y^a(t) = Y(t-T,x_0,u),\quad T\leq t< 2T.
\end{equation}
\item \textbf{(Contamination)} During replay, the adversary contaminates the nominal control signal with a series of malicious control sequences $u^a $ to damage the plant.
\end{enumerate}


{
\begin{remark}
  From the attacker's perspective, one of the main benefits of replay attacks is their simplicity and the possibility of executing them without any system knowledge. For this reason, 
we assume that the attacker has no system knowledge, in particular no knowledge of the watermarking signal. 
\end{remark}}

Before presenting the main detection mechanism of this work, we state the last assumption. We assume that the replay attack occurs when the system, and in particular its observer, has already reached a steady state. Otherwise  the replayed signal will not be consistent with the model and can be easily detected.

\begin{assumption}\label{ass.steady_state}
     The observer \eqref{eq.observer}  has already reached its steady state at time instant $0$. That is, we have   
\begin{equation*} 
\begin{split}
     &\|\widetilde X (t,\bx_0 ,v,y )\| \leq  \alpha_\omega(\|\omega(s)\|_{\infty}) +  \alpha_\nu(\|\nu(s)\|_{\infty}),
\end{split}
\end{equation*}
for all $t\geq 0$, where $\alpha_\omega,\alpha_\nu\in \mathcal{K}$.
\end{assumption}

In this paper, we focus on the replay attack detection problem for nonlinear systems described by   plant \eqref{eq.plant},   observer \eqref{eq.observer} and controller \eqref{eq.controller}. For linear systems, it has been shown that the innovation-based detector is not sufficient for replay attack detection under some stability conditions (see \cite{mo2009secure,mo2013detecting}), and thus the watermark-based detection is employed 
\cite{mo2009secure,mo2013detecting,FANG2020108698,zhao2023stochastic,chen2023replay,naha2023quickest,zhai2021switching,yang2022joint, yu2023reinforcement}.
 Motivated by these works, we consider the following two problems.
\begin{enumerate}
 \item[P1)] For nonlinear systems \eqref{eq.plant},  \eqref{eq.observer} and  \eqref{eq.controller}, is the innovation-based detector sufficient for replay attack detection?
    \item[P2)] If the innovation-based detector is not sufficient, how can we develop a  systematic watermark-based replay detection method for nonlinear systems? 
\end{enumerate}


\section{Innovation-based Replay Attack Detection}\label{sec.innovation_based_detection}
 
In this section, we aim to answer the  P1) by revealing conditions under which the innovation-based detector may fail to detect the replay attack. Specifically, the detector is assumed to have access to the  sensor measurement $y$ (or $y^a$ if under attack) and its estimation $\hat y$, taking the form of
  \begin{equation}\label{eq.detector}
\begin{split}
    &\Sigma_d:
\; \psi = \left\{\begin{aligned}
  &\textnormal{Attack}, \qquad  g(t)  > \vartheta\\
  &\textnormal{No Attack}, ~~    g(t)\leq \vartheta
\end{aligned}\right.\\
& g(t):=\frac{1}{\sigma}  \int_{t-\sigma}^t  \|\widetilde Y(s,   \bx_0, v, y )\|^2 d s
\end{split} 
\end{equation}
where { $g(t)$ is referred to as the monitoring signal of the detector}, $\widetilde Y(s,   \bx_0, v, y )$ is the innovation trajectory given by \eqref{eq.error_observer}, and  $\sigma\in\RR_{>0}$ is the window size.  In the following, we say that the detector is triggered if $\psi = \textnormal{Attack}$.

{ \begin{remark}\itshape
The detector \eqref{eq.detector} is a modification  of  the well-known $\chi^2$ detector \cite{MEHRA1971637,1997A}, which is widely used in CPSs and  takes the form
\begin{equation}\label{eq.detector_chi}
\begin{split}
  & \Sigma_{\chi}:
\; \psi_\chi = \left\{\begin{aligned}
  &\textnormal{Attack}, \qquad  g_{\chi}(t)  > \vartheta_{\chi}\\
  &\textnormal{No Attack}, ~~    g_{\chi}(t)\leq \vartheta_{\chi}
\end{aligned}\right. \\
& g_{\chi}(t)= \frac{1}{\sigma}\int_{t-\sigma}^t  \|\widetilde Y(s,   \bx_0, v, y )\|^2_{\mathcal{P}^{-1}} d s 
\end{split}
\end{equation} 
where $\mathcal{P} $ is the variance of the innovation for normalization and $\vartheta_{\chi  }>0$ is the threshold. For nonlinear systems, the innovation variance is generally nontrivial to compute, and thus we omit the normalization in \eqref{eq.detector_chi} and employ \eqref{eq.detector}.
\end{remark}}
 
 Under the replay attack, at time $t\in [T, 2T)$,  the observer evolves as follows,  
\begin{equation}\label{eq.observer_with_water_fir}\begin{split}
  \dot{\hat { x}} &= f(\hat x)+B (\kappa(\hat x )+v) -L[C\hat x+D(\kappa(\hat x )+v)]  +Ly^a \\
 \hat y  &= C\hat x +D(\kappa(\hat x )+v)
\end{split}
\end{equation}
where $  y^a$ is the replayed data given in \eqref{eq.replay_attack}. We denote by 
$\widehat X(t,\hat x_T,v, y^a)$ and $\widehat Y(t,\hat x_T,v, y^a):= C\widehat X(t,\hat x_T,v, y^a)+D(\kappa(\widehat X(t,\hat x_T,v, y^a))+v)$ the solution and output trajectory of observer 
\eqref{eq.observer_with_water_fir}  with initial state $x_T\in \RR^n$
and subject to the replay attack $y^a$.   In addition, the innovation trajectory under attack is denoted as $\widetilde Y(t,\overbar\bx_0 , v, y^a) := Y(t-T,x_0,u)- \widehat Y(t,\hat  x_T, v, y^a)$, with $\overbar\bx_0 := \text{col}(x_0,\hat  x_T)$.

Then, the following proposition summarizes a scenario where the detector may fail to detect the replay attack.

\begin{proposition}\label{pro.replay}\itshape
   Consider the system  \eqref{eq.plant} with observer   \eqref{eq.observer}, controller  \eqref{eq.controller} with $v(t)=0,\forall t\in \mathbb{R}_{\geq0}$, detector \eqref{eq.detector} and suppose that Assumption \ref{ass.delta_ISS} holds.  Then, if { $\sigma\leq T$ and}
   \eqref{eq.observer_with_water_fir} is $\delta$ISS with input $y^a$ and state $\hat x$, there exists a $\mathcal{K}\mathcal{L}$ function $\overbar \beta$ such that 
\begin{equation}\label{eq.g_t_dISS}
\begin{split}
        g(t)&\! \leq\!   \frac{1}{\sigma} \!  \int^{T}_{t-\sigma} \!  \! \!  \|\widetilde Y(s,   \bx_0, v, y )\|^2 d s\!+\!  \frac{1}{\sigma} \! \int_{T}^t  \! \!   \|\widetilde Y(s\! -\! T,   \bx_0, v, y )\|^2 d s \\
      &+\overbar \beta(\|\hat x(T)-\hat x(0) \|, t) , ~~\quad  t\in[T ,T+\sigma) ,\\
     g(t) &\leq  \! \frac{1}{\sigma} \int_{t-\sigma}^t    \|\widetilde Y(s-T,   \bx_0, v, y )\|^2 d s \! +\! \overbar \beta(\|\hat x(T)-\hat x(0) \|, t), \\
     &\qquad\qquad\qquad\qquad\qquad \qquad  t\in[  T+\sigma,2T).
\end{split}
\end{equation}
\end{proposition}
The proof is postponed to Appendix \ref{prof.fail_detection}.

 Proposition \ref{pro.replay} shows that,    if  $v(t)=0,\forall t\in \mathbb{R}_{\geq0}$ and  
   \eqref{eq.observer_with_water_fir} is $\delta$ISS,
   the detector \eqref{eq.detector} is insufficient for reliably detecting replay attacks. Specifically, {  because all of the innovation terms $\widetilde Y(t,   \bx_0, v, y )$ in \eqref{eq.g_t_dISS} are integrated over  $t\in [0, T )$ and the replay occurs since time instant $T$, }the monitoring signal { $g(t)$} under attack will converge to { some values bounded as in the} {No Attack} scenario. As a result,  the detector may  be triggered only immediately after the replay occurs, and it fails to raise any alarm   as time tends to   infinity. Furthermore, a skilled adversary can potentially avoid even this initial trigger, as demonstrated in the following example.
 
\begin{example}\itshape\label{rem.attack's_intention}
  For systems whose steady state is periodic/quasi-periodic and whose estimator satisfies the $\delta$ISS property given in Proposition \ref{pro.replay}, the adversary could wait until the initial portion of the recorded signal closely matches the current measurements before starting replay. In this case,  $\hat x(T)$ is close to $\hat x(0) $, making the term $\overbar \beta(\|\hat x(T)-\hat x(0) \|, t)$ sufficiently small so that the resulting monitoring signal is indistinguishable (considering sensor noise) from nominal measurements.  
\end{example}

{ We highlight that Proposition~\ref{pro.replay} motivates a detection mechanism for replay attacks.  Indeed, the replay attack can be  detected if \eqref{eq.observer_with_water_fir} satisfies
\begin{equation}\label{eq.exp_expanding}
   \lim_{t\to\infty}\|\widehat X(t,\hat x_T,0,y^a)- \widehat X(t-T,\hat x_0,0,y )\| =\infty.
\end{equation}
Then, it is obvious that during a replay attack the signal $g(t)$ will grow until the detector is triggered. A similar strategy has been explored for instance in \cite{mo2013detecting} for linear systems. However, two issues should be addressed when imposing \eqref{eq.exp_expanding}. First, it opens the possibility of compromising the stability of }{ the whole system by  merely replaying the sensor measurements. Second, designing  feedback gains $K$ and $L$ that enable the system to satisfy  Assumption \ref{ass.delta_ISS} and \eqref{eq.exp_expanding} simultaneously is usually nontrivial in nonlinear cases.}

In view of the previous analysis,  we turn to a watermark-based approach, providing a positive answer to the P2).

 \section{Watermark-based Replay Attack Detection}\label{sec.main}
The watermarking signal is introduced through the term $v$ in controller \eqref{eq.controller}.  
As shown in Fig. \ref{fig.replay_fram22}, it is an additional signal added to the feedback control signal, without requiring a redesign of the controller, observer or detector. 
In this section, { we use the incremental gains introduced in Section \ref{sec.incremental_gain} to evaluate the detection performance and the control system performance loss.  Then, a systematic watermark design approach is proposed based on this evaluation framework.}

 For convenience, the watermark signal is separated into two parts, i.e.,
\begin{equation}\label{eq.watermark}
    v=G\xi(t),
\end{equation}
 where $G \in\mathbb{R}^{m\times m}$ is the weight (gain) matrix to be designed and $\xi(t)$ is a Lebesgue  integrable watermark signal    satisfying
\begin{equation}\label{eq.watermark_bound}
    \frac{1}{t } \int_{t_0}^{t_0+t}\|\xi(s )\|^2ds  \leq  1, \qquad\forall t> t_0.
 \end{equation}
 Denote $ \mathcal{U}_\xi$ the set of $\xi$ that satisfying \eqref{eq.watermark_bound}.

With the watermark given in \eqref{eq.watermark}, the attacked observer \eqref{eq.observer_with_water_fir}  becomes \begin{equation}\label{eq.observer_with_water}\begin{split}
  \dot{\hat { x}} &= f(\hat x )+B( \kappa(\hat x )+G\xi)  - L[C\hat x\!+\!D(\kappa(\hat x )  + G\xi)] \! +\!Ly^a \\
 \hat y  &= C\hat x +D(\kappa(\hat x )+G\xi).
\end{split}
\end{equation}

\begin{figure}[t]
\centering
\includegraphics[width=1\columnwidth]{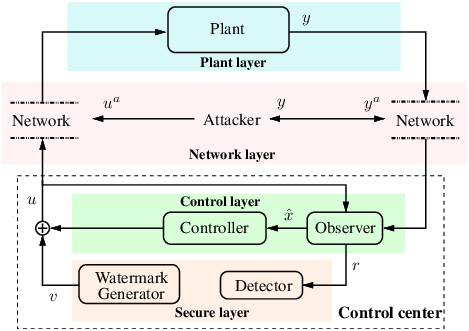}
\caption{Scheme of the communication topology and security layer.}
\label{fig.replay_fram22}
\end{figure}

\begin{remark}\itshape
Watermark design is well studied for linear systems, but it is not straightforward to extend these methods to the nonlinear case.
    For linear systems, one of the most widely used watermarks is  independent and identically distributed (i.i.d.) Gaussian noise \cite{mo2009secure,zhai2021switching,FANG2020108698,zhao2023stochastic,chen2023replay}, i.e.,
$ v \sim \mathcal{N}(0,\Sigma) $.
 With $v$, the covariance of the innovation under a replay attack differs from the healthy case, and this difference can be calculated precisely  to evaluate the detection performance. In addition, the performance loss caused by the watermark can also be calculated precisely  by extra watermark-induced LQG performance loss. However, when the plant is nonlinear, {the exact evaluation of these two factors is much more complex}. Therefore, { appropriate} indices and, correspondingly, a new framework are required, which motivates the present work.
\end{remark}

\subsection{Replay Attack Detection Performance}\label{subsec.detection_perform}

An ideal detection mechanism will never trigger an alarm in the absence of an attack (no false positives) and must trigger an alarm under the replay attack (no false negatives). 
{Under  Assumption 1, the ISS property of the error dynamics holds  independent of the inputs. Therefore, in the absence of an attack, even if the watermarking signal changes the statistics of the innovation, the upper bound on the estimation error, and hence the upper bound on $g(t)$, remains unchanged. In addition, under the replay attack,  the difference between the current and historical watermark employed by the estimator guarantees a lower bound on $g(t)$,  thus ensuring attack detection.}
Hence, to quantify the detection performance of the system, the following theorem presents the upper bound on $g(t)$ in the absence of an attack and the lower bound on $g(t)$ under the replay attack. 
{ Additionally, this result can be used to design the threshold $\vartheta$ in the detector \eqref{eq.detector}, to avoid false positives induced by the noise.  }


\begin{theorem}\label{theorem.output_difference}
    Consider  system   \eqref{eq.plant} with observer   \eqref{eq.observer}, controller \eqref{eq.controller}, detector  \eqref{eq.detector} and let Assumptions~\ref{ass.delta_ISS} and \ref{ass.steady_state} hold. In addition,  suppose that system \eqref{eq.observer_with_water} {has an $[\mathcal{L}^-_{\delta 2}]^\xi_{\hat y}$ gain as defined in Definition~\ref{def.L_minus}}. Then, in the absence of replay attack, for all $t\geq 0$,
   \begin{equation}\label{eq.dete_no_attack_pro}
   \begin{split}
       & g(t) \leq \Big[ \|C\|  (\alpha_\omega (\overbar \omega)+\alpha_\nu(\overbar \nu))+\overbar \nu \Big]^2=:g_n,     
   \end{split}
    \end{equation}
    where  $\overbar \omega$ and $\overbar \nu$ are the upper bounds of $\|\omega\|_\infty$ and $\|\nu\|_\infty$, respectively, $\alpha_\omega$ and $\alpha_\nu$  are   given in Assumption \ref{ass.steady_state}.
 
     In the presence of a replay attack launched at $T$, for all $t\in[ T,2T)$,
\begin{equation}\label{eq.dete_with_attack_pro}
\begin{split}
 g(t)&\geq    \frac{[\mathcal{L}^-_{\delta 2}]^\xi_{\hat y}-\varepsilon }{\sigma}  \int^{ t}_{\Theta}   \|\xi(s)-\xi(s-T)\|^2 d s\\
  &-\frac{\overbar \alpha^-_a(\|\hat x(\Theta)-\hat x(\Theta-T)\|)}{\sigma} - g_n=:g_{a}(t),
\end{split}
\end{equation}
 where $ \varepsilon >0 $ is an arbitrarily small constant, $\Theta=T$ if $ T\leq t<T+\sigma$ and $\Theta=t-\sigma$ if $t\in[  T+\sigma,2T)$.
\end{theorem} 
 
The proof is postponed to Appendix \ref{Appen.output_difference}.

With Proposition \ref{theorem.output_difference}, we  obtain the following corollary.
\begin{corollary}\label{coro.dete}\itshape
Under the conditions in Theorem \ref{theorem.output_difference}, the 
detector will have no false positive alarms for $t\geq 0$ and no false negative alarms for $t\in[ T,2T)$ if
 \begin{equation}\label{eq.dete_condition}
     \begin{split}
 g_n < g_a(t),~~\forall t\in[ T,2T)
     \end{split}
\end{equation}
and the threshold $\vartheta$ is chosen such that\begin{equation}\label{eq.no_false_alarm}
     \begin{split}
 g_n\leq \vartheta<g_a(t),~~\forall t\in[ T,2T).
     \end{split}
\end{equation}

\end{corollary}


{
\begin{remark}
    Under the replay attack, the sensor measurements are { substituted} by the previously recorded healthy data, so any modifications to the control input introduced by the attacker are not reflected in the sensor measurements received by the control center. Hence, the attack input signal does not appear in either the detection mechanism or   the detection performance analysis.
\end{remark}}

  {\begin{remark}\label{remark.initial}\itshape
  The term $-(1/{\sigma}) \overbar \alpha^-_a(\|\hat x(\Theta)-\hat x(\Theta-T)\|) $ in \eqref{eq.dete_with_attack_pro} is negative because, in the worst case, the initial state difference, i.e. $\hat x(\Theta)-\hat x(\Theta-T)$, may cancel out the effect of the watermark signal   on the detector's output. However, {in most cases, this term will be small} over  $t\in [T,T+\sigma)$, because in  this period   $\hat x(\Theta)-\hat x(\Theta-T) = \hat x(T)-\hat x(0)$, which is typically chosen to be small by the adversary  to avoid   triggering the detector immediately, as analyzed in Section \ref{sec.innovation_based_detection}. 
\end{remark}}

\subsection{ Control System Performance Loss}\label{subsec.limit_perform}
Since the watermark is essentially a disturbance acting on the physical plant, it {inevitably} degrades the performance of the closed-loop control system. In this subsection, this degradation is evaluated by comparing the output difference of the physical plant with and without the watermark.

With the watermark, by \eqref{eq.plant}-\eqref{eq.controller}, the observer and nonlinear plant   evolve as \begin{equation}\label{eq.perform_with_watermark_full}
\begin{split}
    \dot{\hat x}&= f(\hat x)+B(\kappa(\hat{ x})+G\xi)+ L(Cx-C\hat x + \nu) \\
    \dot{x}&= f(x)+B(\kappa(\hat{ x})+G\xi ) +\omega\\
  y  &= Cx+D(\kappa(\hat{ x})+G\xi  )+\nu.
\end{split}
 \end{equation}

 By  Assumption~\ref{ass.delta_ISS},   the effect of the watermark on the trajectories of the closed-loop system  is bounded. In this subsection, we provide a bound on the degradation via an $\mathcal{L}^+_{\delta 2}$ gain, to quantify how different  the system response becomes once the watermark is implemented.

\begin{theorem}\label{thm.per_loss}
Suppose that  system \eqref{eq.perform_with_watermark_full} has an $[\mathcal{L}^+_{\delta 2}]^\xi_{\hat y}$ gain as defined in Definition~\ref{def.L_ifty}. Then for any given time instant $\tau> 0$, 
the following bound holds

\begin{equation}\label{eq.L_infty_gain}
\begin{split}
    \|( Y(t,x_{\bar t_0}, \xi)- Y(t,x_{\bar t_0}, 0)_\tau\|_{\mathcal{L}_2}  & \leq ([\mathcal{L}^+_{\delta 2}]^\xi_{y} +\varepsilon')(\tau -\overline t_0),
\end{split}
\end{equation}
 where  
 $\varepsilon'>0$ is an arbitrary small constant.
\end{theorem}

The proof of Theorem \ref{thm.per_loss} can be obtained by comparing \eqref{eq.perform_with_watermark_full} with and without the watermark, along with the fact that $\xi(t)$ is Lebesgue-integrable and satisfies 
\eqref{eq.watermark_bound}.

To simplify the analysis and the design procedure, for the rest of the paper, we will not consider the effect of the estimation error   $\tilde x$ when  computing the control system performance loss,  i.e., instead of using \eqref{eq.perform_with_watermark_full},  the following approximated  system is adopted 
\begin{equation}\label{eq.perform_with_watermark2}
\begin{split}
  \dot{x}'&= f(x')+B(\kappa(x' )+G\xi) +\omega\\
  y'  &= Cx' +D(\kappa(x' )+G\xi  )+\nu.
\end{split}
\end{equation}
Since the estimation error is ISS by means of Assumption~\ref{ass.delta_ISS} independently on the watermarking signal, this simplification is reasonable if  $l_{\kappa} $ is relatively small.

\color{black}

\subsection{Two Examples of Watermark Signal $\xi(t)$}
This subsection gives two examples for the design of the  watermark signal $\xi(t)$, and the watermark gain $G$ will be optimized in the next subsection.

{The watermarking signal must satisfy  the  Lebesgue
 integrability of  $\xi(t)$ and \eqref{eq.watermark_bound} in order to ensure that the system performance loss is bounded as \eqref{eq.L_infty_gain}.} In addition, to achieve \eqref{eq.dete_condition}, a relatively large $\int^{ t}_{\Theta} \|\xi(s)-\xi(s-T)\|^2 d s$ is required. Therefore, for some  $s\in[\Theta,t ]$, we need that
$\xi(s)\neq \xi(s- T )$,   
implying that $\xi(t)$ cannot be a periodic signal. We highlight that quasi-periodic signals may also be undesired as   $\int^{ t}_{\Theta} \|\xi(s)-\xi(s-T)\|^2 d s$ is  small {if $\xi(s)$ is close to $\xi(s-T)$ for all $s\in[\Theta,t]$.} 

Following the above analysis, we give two desirable examples of $\xi(t)$.

\subsubsection{Chaotic Watermark}

Chaotic systems are deterministic systems, but appear random and non-periodic \cite{baker1996chaotic}.
The chaotic signal $\xi(t)$ can be generated by a chaotic system of the following form
\begin{equation}\label{eq.chaotic_system}
  \left\{\begin{aligned}
  \dot{\theta}&= A\theta+\phi(\theta)\\
  \xi &= \Lambda {\theta}
\end{aligned}\right.
\end{equation}
where $\theta\in \overbar\Theta \subset \mathbb{R}^{n_\xi}$ with $\overbar\Theta$ a compact set, $\xi\in \mathcal{U}_\xi \subset\mathbb{R}^{m}$,   $\phi:\RR^{n_\xi}\to\RR^{n_\xi}$ is a nonlinear vector field. Since $\xi(t)$ is continuous, it is Lebesgue  integrable  and  \eqref{eq.watermark_bound} can be achieved by adjusting the output matrix $\Lambda$. 
In addition, since the chaotic signal is non-periodic, it is clear $\int^{ t}_{ \Theta   }  \|\xi(s) - \xi(s-T)\|^2  d s>0$ for all $t>\Theta$. Moreover, even if $\xi(T)$ is close to $\xi(t-T)$, $\xi(t)$ and $\xi(t-T)$ may deviate from each other exponentially \cite[Chapter 26]{greiner2003classical}.

\subsubsection{Bernoulli Watermark}
The watermark can also be generated by stochastic systems. Here we consider  a simple watermark generated by Bernoulli distribution.

For $i = 1,2,..,m$, let the watermark be, for $t\in [t_i,t_i+\delta_t)$,
\begin{equation}\label{equ.bern_water} 
 \xi_{i}(t)= \left\{\begin{aligned}
  \frac{1}{\sqrt{m} },\quad  \varrho( t_i) = 1,\\
  -\frac{1}{\sqrt{m}},\quad\varrho( t_i) =0
\end{aligned}\right.
\end{equation}
where   $\{\varrho(t_i)\}$ is a series of random variables and each of them  follows Bernoulli distribution with ${\mathbb{P}}(\varrho(t_i) = 1) ={\mathbb{P}}(\varrho(t_i) = 0) =0.5$, and $ \mathbb{E}(\varrho( t_i)\varrho( t_j))=0, \forall i\neq j$. With fixed $\delta_t>0$, $\xi(t)$ is Lebesgue  integrable  and it is easy to verify that this watermark satisfies \eqref{eq.watermark_bound}. In addition, with this watermark, $ \int^{t}_{\Theta}  ( \|\xi(s)-\xi(s-T)\|^2)d s >0$ can be guaranteed with high probability if the term $\sigma$ is selected sufficiently large.

 { The main difference between watermark signals lies in how they guarantee a sufficiently large $ \int^{ t}_{\Theta}   \|\xi(s)-\xi(s-T)\|^2 d s$. For example, the detection performance of the chaotic watermark  is guaranteed by the inherent properties of chaotic systems, e.g., non-periodic and positive Lyapunov exponents.
On the other hand, for the Bernoulli watermark, the detection performance is guaranteed   in a probabilistic sense.  }

\section{Design of the Watermark Gain}\label{sec.solving_G}
{ In this section, we aim to design the watermark gain $G$. According to Theorems \ref{theorem.output_difference} and \ref{thm.per_loss}, a relatively large value of $[\mathcal{L}^{-}_{\delta 2}]^{\xi}_{\hat y}$ for \eqref{eq.observer_with_water} implies a better   detection performance, while a relatively small value of $[\mathcal{L}^+_{\delta 2}]^\xi_{y'}$ for \eqref{eq.perform_with_watermark2} corresponds to a less degradation in control system performance.  From \eqref{eq.observer_with_water} and \eqref{eq.perform_with_watermark2}, it is clear that both incremental gains   depend on the watermark gain   $G$. To show this relationship, we denote these two incremental gains as  $[\mathcal{L}^{-}_{\delta 2}]^{\xi}_{\hat y}(G)$ and $[\mathcal{L}^+_{\delta 2}]^\xi_{y'}(G)$, respectively. In the following, we proceed to design the watermark  gain $G$  to achieve a balance  between detection performance and control system performance loss.} A feasible balance is to  maximize the $[\mathcal{L}^{-}_{\delta 2}]^{\xi}_{\hat y}(G) $ gain while constraining the $[\mathcal{L}^+_{\delta 2}]^\xi_{y'}(G) $ gain.  Consequently, the following optimization problem is constructed,
\begin{equation}\label{eq.opti_water2}
    \begin{split}
  \max_{{G},\beta}~~&\beta \\
 {\rm s.t.~~~}& [\mathcal{L}^-_{\delta 2}]^\xi_{\hat y}(G)  \geq \beta  \text{~and~} [\mathcal{L}^+_{\delta 2}]^\xi_{y'}(G)  \leq \alpha
 \end{split}
\end{equation}
for  a given $\alpha>0$ which bounds the allowable performance loss.

The detailed   process for solving \eqref{eq.opti_water2} is given in the following. {Deriving tractable matrix inequalities in order to solve \eqref{eq.opti_water2} and solving the resulting bilinear matrix inequalities are two   main difficulties in this section.} For simplicity,  we   assume that $\kappa(\hat x) = -K\hat x$.

\subsection{Sufficient Conditions for Incremental Gains}
 For preparation, a sufficient condition such that $ [\mathcal{L}^-_{\delta 2}]^\xi_{\hat y}(G)\geq \beta$ holds for a given $\beta$ is given below.  This condition will be used later to develop a design methodology that optimizes the detection performance.

\begin{proposition}\label{propo.L_minus}
   For system \eqref{eq.observer_with_water}, the bound $ [\mathcal{L}^-_{\delta 2}]^\xi_{\hat y}(G) \geq \beta$ holds for some positive constant $\beta>0$ if    
    there exists a symmetric negative definite matrix $Q\in\RR^{n\times n}$ and matrix $G\in\RR^{m\times m}$ such that, for all {$x\in \mathbb{R}^n$},
       \begin{equation}\label{eq.LMI_L_minus}
   \begin{bmatrix}
           \mathcal{M}_{11}&  Q(B   -LD )G+(C -D  K)^\top  D   G\\
       \star &G^\top  D  ^\top    D    G-\beta I_m \\
       \end{bmatrix} \succeq 0
       \end{equation}
where $ \mathcal{M}_{11}=   [A_x -B   K-L(C  -D   K)]^\top   Q+  Q[A_x -B   K-L(C  -D   K)]+(C -D  K)^\top  (C -D  K) $.
\end{proposition}
 The proof is obtained by directly applying Lemma \ref{lem.L_minus_ori}.

Below, we provide a sufficient  condition such that $ [\mathcal{L}^+_{\delta 2}]^\xi_{  y}(G) \leq \alpha$ for some given $\alpha$. This condition will be fundamental for developing a methodology that { constraints} the system performance loss. 

\begin{proposition}\label{propo.L_infinity2}
   For system \eqref{eq.perform_with_watermark2}, the bound $[\mathcal{L}^+_{\delta 2}]^\xi_{y'}(G)  \leq \alpha$ holds for some positive constant $\alpha>0$ if there exists a symmetric positive definite matrix $P_s\in \RR^{n\times n}$ and matrix $G\in\RR^{m\times m}$ such that, for all {$x\in \mathbb{R}^n$},
\begin{equation}\label{eq.LMI_L_infinityG}
    \begin{bmatrix}
       \mathcal{N}_{11} &B   G&P_s C ^\top - PK^\top D  ^\top \\
        \star&-\alpha I_m&(D  G)^\top\\
       \star&\star&- I_p
    \end{bmatrix}\preceq 0
\end{equation}
where $\mathcal{N}_{11}: =  A_x P_s+P_s A_x^{\top} -  B  KP_s-P_s K^{\top} B^{\top} +\epsilon P_s$ with $\epsilon>0$.
\end{proposition}
The proof is postponed to Appendix \ref{proof.L_infinity}.

\begin{remark}\itshape
 {  Due to the existence of state   dependent differential  terms, i.e., $A_x$, { an infinite set} of  matrix inequalities should be {considered} when solving    \eqref{eq.LMI_L_minus} and \eqref{eq.LMI_L_infinityG}.  
  One way to deal with this problem is to convexify the state  variation,  referred to as the differential parameter-varying inclusion   in \cite[Sec. 4.5]{verhoek2023convex}. With this convexification,   we are able to solve \eqref{eq.LMI_L_minus} and \eqref{eq.LMI_L_infinityG} with a finite number of matrix inequalities via polytopic or multiplier-based methods \cite{hoffmann2014survey}.} 
Alternatively, a finite set of LMIs can be obtained for particular partially linear systems. { For example, one can refer to \cite{zoboli2024dynamic} for the case where $A_x= A + \phi(x)$ with $\phi(x)$ satisfying some differential quadratic constraint.}
\end{remark}

\subsection{An Algorithm for   Watermark Gain Design}
For solving the   optimization problem \eqref{eq.opti_water2}, the sufficient conditions for $[\mathcal{L}^-_{\delta 2}]^\xi_{\hat y}(G) \geq \beta$ in Proposition \ref{propo.L_minus} and for $[\mathcal{L}^+_{\delta 2}]^\xi_{y'}(G) \leq \alpha$ in Proposition \ref{propo.L_infinity2} should be used. It is noted that due to the presence of the bilinear terms $Q(B   G-LD  G)$ and $G^\top  D  ^\top   D   G$, \eqref{eq.LMI_L_minus} is not an LMI. To solve this problem, we resort to the iterative LMI technique \cite{cao1999simultaneous}, where an equivalent condition to \eqref{eq.LMI_L_minus} is utilized, as presented in the following proposition.
 \begin{proposition}\label{pro.G_minus}
   Condition \eqref{eq.LMI_L_minus} holds if and only if  there exist  symmetric negative definite matrices $Q_0, Q\in\RR^{n\times n}$ and matrices $G_0, G \in\RR^{m\times m}$ such that, for all {$x\in \mathbb{R}^n$}
       \begin{equation}\label{eq.LMI_L_minusG}
   \begin{bmatrix}
           \mathcal{M}'_{11}&  ( C - D   K)^\top D   G&-Q \\
      \star& \mathcal{M}'_{22} &(B   G-LD  G)^\top\\
      \star&\star&I
       \end{bmatrix} \succeq 0 
       \end{equation}
   where    $ \mathcal{M}'_{11}= QQ_0+Q_0Q-Q_0Q_0+  Q A_x+  A_x^\top   Q- Q    B  K   - (B   K)^\top  Q- QLC  - (LC )^\top Q +QLD  K+(LD  K)^\top Q + ( C - D  K)^\top  (C - D   K) $, $\mathcal{M}'_{22}= G^\top[(B   -LD   )^\top(B  -LD   ) +   D  ^\top   D  ] G_0+G_0^\top[(B   -LD   )^\top(B   -LD   ) +   D  ^\top   D  ] G-G_0^\top[(B -LD   )^\top(B  -LD   ) +   D  ^\top   D  ] G_0-\beta I_m$.
       \end{proposition}
       The proof is postponed to Appendix \ref{appendix.prof_G_minus}.

It is clear that \eqref{eq.LMI_L_minusG} is an LMI for any given $Q_0$ and $G_0$.    Then \eqref{eq.LMI_L_infinityG} and \eqref{eq.LMI_L_minusG} can be used to replace  $[\mathcal{L}^-_{\delta 2}]^\xi_{\hat y}(G) \geq \beta$ and $[\mathcal{L}^+_{\delta 2}]^\xi_{y'}(G)  \leq \alpha$ in \eqref{eq.opti_water2}, respectively. It is noted that with this replacement, we may only obtain a sub-optimal solution, as \eqref{eq.LMI_L_infinityG} and \eqref{eq.LMI_L_minusG} are not guaranteed to be the necessary conditions of  $[\mathcal{L}^-_{\delta 2}]^\xi_{\hat y}(G)  \geq \beta$ and $[\mathcal{L}^+_{\delta 2}]^\xi_{y'}(G) \leq \alpha$, respectively.  Specifically, an iterating process for solving $G$ is given in Algorithm 1. 
\begin{algorithm}  
{\textbf{Algorithm 1:} Solving the watermark gain $G$}
\begin{algorithmic}[1]
\Require
Parameters $A_x,B  ,C ,D  $, $K, L$, $G_\text{init}$, {$Q_\text{init}$}, $\alpha$, $\beta_0$, number of iterations $N$
\Ensure
 $G_\text{opt}$;
 \State Set $G_0  =G_\text{init}$, { $Q_0  =Q_\text{init}$}, $i =1$ 
\While{$i\leq N$}  
                \State With the obtained $G_0 $, {$Q_0$}, maximize $\beta_i$ subject to \eqref{eq.LMI_L_infinityG} and \eqref{eq.LMI_L_minusG}    to get $G$, {$Q$}.   Let $G_0= G$, {$Q_0 = Q$}.  
                \State $i=i+1$.
            \EndWhile\label{euclidendwhile}
\end{algorithmic}
\end{algorithm}

   \begin{remark}\itshape\label{remark_converge_beta}
       In   Algorithm 1,  $G_\text{init}$ and {$Q_\text{init}\prec 0$}  can be  any matrices    such that the maximization at Step 3 is solvable. Then, at each iteration of Step 3, the LMIs \eqref{eq.LMI_L_infinityG} and \eqref{eq.LMI_L_minusG} are always solvable since the optimization variables can always take their initial values (in this case $G=G_0$, {$Q=Q_0$} and $\beta_{i+1} = \beta_i$), implying that $\beta_i$ increases monotonically.
{Meanwhile, according to   \eqref{eq.LMI_L_infinityG} and applying the Schur complement lemma to    \eqref{eq.LMI_L_minusG}, we have
\begin{equation*}
   \beta_i I_m\preceq  (D  G)^\top D  G\preceq \alpha I_m,
\end{equation*}
ensuring that  $\beta_i$ is bounded.  Then combining the fact that  $\beta_i$ is monotonically increasing and bounded, it is convergent. However, $\beta_i$ is guaranteed to converge to a local maximum,  denoted as  $\beta_{opt}$, rather than to the global one.}
{Finally,   with $\beta_i$ converging  to $\beta_{opt}$, $G$ will also converge to the set  
\begin{equation*}
   \mathcal{G} = \{G \ | \ \eqref{eq.LMI_L_infinityG}   \text{ and } \eqref{eq.LMI_L_minusG} \text{ with } \beta = \beta_{opt}\}.
\end{equation*}   }
   \end{remark}

\subsection{Co-design of Watermark, Controller and Observer}\label{sec.co-design}
{ From dynamics \eqref{eq.observer_with_water} and \eqref{eq.perform_with_watermark2} and matrix inequalities  \eqref{eq.LMI_L_minus} and \eqref{eq.LMI_L_infinityG}, it is clear that the parameters of the observer $L$, and controller $K$, also affect the incremental gains of the system, which are therefore denoted   by $[\mathcal{L}^-_{\delta 2}]^\xi_{\hat y}(G,L)$ and $[\mathcal{L}^+_{\delta 2}]^\xi_{y'}(G,K,L) $.
Hence, to gain a better balance between detection performance and control system performance loss}, we  extend our main results to co-design $G, K$ and $L$.

When designing $K$ and $L$, we must ensure that Assumption \ref{ass.delta_ISS} is satisfied, as it guarantees that the plant and the error dynamics have $\delta$ISS and ISS properties, respectively. Consequently, the optimization problem \eqref{eq.opti_water2} is extended to 
\begin{equation}\label{eq.opti_main}
    \begin{split}
  \max_{G,K,L,\beta}~~&\beta  \\
 {\rm s.t.~~~}&  [\mathcal{L}^-_{\delta 2}]^\xi_{\hat y}(G,L)\geq  \beta, ~~[\mathcal{L}^+_{\delta 2}]^\xi_{y'}(G,K,L)  \leq \alpha\\
 &\qquad\qquad\qquad \text{and Assumption }   \ref{ass.delta_ISS} \text{ holds}
 \end{split}
\end{equation}
for some given $\alpha $.

To solve \eqref{eq.opti_main}, a sufficient condition to achieve  Assumption   \ref{ass.delta_ISS} is first provided below.
\begin{figure*}[hb]
	\centering
	\vspace*{8pt}
	\hrulefill
	\vspace*{8pt} 
	\begin{eqnarray*}\setlength{\arraycolsep}{10.5pt}
  \Xi:= \begin{bmatrix}
   \mathcal{M}''_{11}& \star & \star &\star&\star&\star&\star&\star&\star&\star&\star  \\
      ( D   G)^\top C &\mathcal{M}''_{22}&\star&\star&\star&\star&\star&\star&\star&\star&\star\\
       Q +     L  C  &0&I_n&\star&\star&\star&\star&\star&\star&\star&\star\\
      Q +   B   K &0&0&I_n&\star&\star&\star&\star&\star& \star&\star\\
     - Q& B  G&0&0&I_n&\star&\star&\star&\star&\star&\star\\
        D   K&D   G&0&0&0&I_m&\star&\star&\star&\star&\star\\
         -D  K&0&0&0&0&0& \mathcal{M}''_{77}  &\star&\star&\star&\star\\
           Q&0&0&0&0&0&-L&I_n&\star&\star &\star\\
             0&D  G&0&0&0&0&0&0&\mathcal{M}''_{99}&\star&\star\\
              Q&0&0&0&0&0&0&0&-L&I_n&\star\\
              2Q- L(Q_0 L_0)^\top \!\!\!\! \!\!\!\!\!\!\!\! &  0&0&0&0&0&0&0&0&0&I_n
    \end{bmatrix} 
	\end{eqnarray*}
    where $ \mathcal{M}''_{11}=  Q A_x +A_x^\top Q+(B  K)^\top B  K_0+(B  K_0)^\top B  K-(B  K_0)^\top B  K_0
  + 9Q Q_0+9Q_0 Q-9Q_0 Q_0+(L C )^\top   L_0 C  +(L_0 C )^\top   L C -(L_0 C )^\top   L_0 C +2(   D   K)^\top    D   K_0+2(   D   K_0)^\top    D   K-2(   D   K_0)^\top    D   K_0-2Q_0L_0(Q_0L_0)^\top+(Q_0 L_0)L^\top L_0 (Q_0 L_0)^\top+(Q_0 L_0)L_0^\top L (Q_0 L_0)^\top-(Q_0 L_0)L_0^\top L_0 (Q_0 L_0)^\top$, $\mathcal{M}''_{22}=3G^\top     D  ^\top     D   G_0+3G_0^\top     D  ^\top     D   G-3G_0^\top     D  ^\top     D   G_0+G^\top     B^\top     B  G_0+G_0^\top     B ^\top     B  G-G_0^\top     B^\top     B  G_0-\beta I_m $, $\mathcal{M}''_{77} =  I_m+L^\top L_0+L_0^\top L-L_0^\top L_0$, and $\mathcal{M}''_{99}=  I_m+L^\top L_0+L_0^\top L-L_0^\top L_0$.
\end{figure*} 
\begin{lemma}\label{lem.contraction_1}
    Consider the system  \eqref{eq.plant} with observer   \eqref{eq.observer}, controller  \eqref{eq.controller}. Assumption \ref{ass.delta_ISS} holds with  $\kappa(\hat x)= -K\hat x$ if there exist  positive definite matrices 
   {$R$, $S$, matrices  $L$ and $K$, and real constants $\epsilon_1,\epsilon_2$}  such that, for all {$x\in \mathbb{R}^n$},
   \begin{subequations}
   \begin{align}\label{eq.obsv2}
   A_x^\top  R+R    A_x {- C ^\top L^\top R-RL C} &\preceq -\epsilon_1 R\\
   \label{eq.ctrl2}
        A_x S+S A_x^{\top} -  B  KS-S K^{\top} B ^{\top}&\preceq -\epsilon_2 S  
   \end{align}
   \end{subequations}
\end{lemma}

The proof is postponed to Appendix~\ref{proof.contra}.

Now we can use  \eqref{eq.LMI_L_minus}, \eqref{eq.LMI_L_infinityG}, \eqref{eq.obsv2} and \eqref{eq.ctrl2} for solving \eqref{eq.opti_main}. It is noted that the condition \eqref{eq.ctrl2} has already been involved in $\mathcal{N}_{11}\preceq 0$ in  \eqref{eq.LMI_L_infinityG}, and thus it is omitted  in the following design process. Then similar to the iterating process to \eqref{eq.opti_water2}, the iterative LMI technique should be employed to solve \eqref{eq.opti_main}.
To do so,   we first give equivalent conditions to  \eqref{eq.LMI_L_minus}, \eqref{eq.LMI_L_infinityG}  and \eqref{eq.obsv2}, respectively. Their proofs are similar to the proof of Proposition \ref{pro.G_minus}, and thus omitted here.

\subsubsection{An Equivalent Condition to \eqref{eq.LMI_L_minus} }\label{prof.LMI_H_inf_equva}

\begin{proposition} \label{pro.equva_1}
   Condition \eqref{eq.LMI_L_minus} holds if and only if there exist symmetric negative definite matrices $Q_0, Q $ and matrices $K_0, L_0, G_0, K, L$ and $G$ such that
 \begin{equation}\label{eq.ite_LMI_L_minus}
     \Xi\succeq 0
 \end{equation}
  for all {$x\in \mathbb{R}^n$}, where $\Xi$ is given at the bottom of this page.
\end{proposition}
 
\begin{algorithm}[tb]
{\textbf{Algorithm 2:} Solving $K$, $L$, and $G$}
\begin{algorithmic}[1]
\Require
Parameters $A_x,B ,C ,D  $, $\alpha$, $\beta_0$, $G_\text{init}$, number of iterations $N$
\Ensure
 $K_\text{opt}$, $L_\text{opt}$, and $G_\text{opt}$;
 \State Set $G  =G_\text{init}$ and let $G_0=G$.  Solve \eqref{eq.LMI_L_infinityG} to get $K$ and $P_s$, and let $K_0=K$, $P_{s0}=P_s$. Solve \eqref{eq.obsv2} to get $L$ and $R$, and let $L_0=L$, $R_0=Q$.

 \State With $G , K $ and $L$ obtained at Step 1, solve \eqref{eq.LMI_L_minus} to get $Q$ and let $Q_0=Q$.
 \State Set $i =1$.
\While{$i\leq N$}  
                \State With the obtained $G_0, K_0$, $L_0$, $P_{s0}$, $Q_0$ and $R_0$, maximize $\beta_i$ subject to  \eqref{eq.ite_LMI_L_minus}, \eqref{eq.LMI_L_infinity2}, and  \eqref{eq.ite_LMI_L_infty} to get $G_\text{opt}, K_\text{opt}$, $L_\text{opt}$, $P_{\text{opt}}$, $Q_\text{opt}$ and $R_\text{opt}$. Let $G_0= G_\text{opt}, K_0=K_\text{opt}$, $L_0=L_\text{opt}$, $P_{s0} =P_{s\text{opt}}$, $Q_0 = Q_\text{opt}$ and $R_0 =R_\text{opt}$.
                \State $i=i+1$.
            \EndWhile\label{euclidendwhile2}
\end{algorithmic}
\end{algorithm}

\subsubsection{An Equivalent Condition to \eqref{eq.LMI_L_infinityG}}\label{Appen.LMI_equva_H_inf}

\begin{proposition}\label{pro.equva_2}
   Condition \eqref{eq.LMI_L_infinityG} holds if and only if there exist symmetric positive definite matrices  $P_{s0}$, $P_s $  and matrices $ K_0$,  $ K$  such that, for all {$x\in \mathbb{R}^n$},
\begin{equation}\label{eq.LMI_L_infinity2}
    \begin{bmatrix}
    \mathcal{N}_{11}' &B  G&P_s C ^\top  &P_s -B  K&P_s \\
       \star&-\alpha I_m&(D  G)^\top&0&0\\
        \star&\star&\mathcal{N}_{33}'  &0&-D  K\\
        \star&\star&\star&-I_n&0\\
        \star&\star&\star&\star&-I_n
        \end{bmatrix} \preceq 0
\end{equation}
where $\mathcal{N}_{11}' =    A_x P_s+ P_s  A_x^\top - 2P_{s0} P_s- 2P_s P_{s0}+ 2P_{s0}P_{s0}-   B  K_0(   B   K)^\top -  B   K(   B  K_0)^\top +  B   K_0(   B   K_0)^\top $ and $\mathcal{N}_{33}' = -I_p -D  K_0(D  K)^\top -D  K(D  K_0)^\top+ D  K_0(D  K_0)^\top$.

\end{proposition}

\subsubsection{An Equivalent Condition to  \eqref{eq.obsv2}} 
\begin{proposition}\label{pro.equva_3}

Condition \eqref{eq.obsv2} holds if and only if there exist symmetric positive definite matrices $R_{0} $, $R $   and matrices $L_0$, $L$ such that, for all {$x\in \mathbb{R}^n$},
\begin{equation}\label{eq.ite_LMI_L_infty}
  \begin{bmatrix}
      \mathcal{U}_{11}&R-(LC )^\top \\
     \star&-I_n
  \end{bmatrix} \preceq 0
\end{equation}
where $ \mathcal{U}_{11} =  A_x^\top R+R A_x-R_{0}R-RR_{0}+R_{0}R_{0}-( L_0 C )^\top  L C -( L C )^\top  L_0C +(L_0C )^\top L_0 C $.

\end{proposition}
 
Then, with  $[\mathcal{L}^-_{\delta 2}]^\xi_{\hat y}(G,L)\geq  \beta$, $[\mathcal{L}^+_{\delta 2}]^\xi_{y'} (G,K,L) \leq \alpha$  and \eqref{eq.obsv2} in the optimization problem \eqref{eq.opti_main} replaced by \eqref{eq.ite_LMI_L_minus}, \eqref{eq.LMI_L_infinity2} and \eqref{eq.ite_LMI_L_infty}, respectively, the iterating process to $K$, $L$ and $G$ is given in Algorithm 2.

\begin{remark}\itshape
    With a similar analysis to Algorithm 1 (in Remark \ref{remark_converge_beta}),   $\beta_i$ will converge to a local maximum. 
\end{remark}
 
 \begin{remark}\itshape
     According to Proposition \ref{theorem.output_difference}, the detection performance can also be improved by reducing $g_n$
     if $K$ and $L$ can be designed. In \eqref{eq.opti_main} we still focus on improving $g_a(t)$ for simplicity.
 \end{remark}


\section{Special Case: Linear Systems}\label{sec.case}

 In this section, we consider the case where the plant is linear. In particular, we consider the following plant
\begin{equation*} 
\Sigma_{lp}:
\; \left\{
\begin{split}
  \dot{x}&= Ax+Bu+\omega,\\
  y  &= C{x}+Du+\nu,
\end{split}
\right.
 \end{equation*}
where $\omega\sim \mathcal{N}(0,\mathcal{Q})$ and $\nu\sim \mathcal{N}(0,\mathcal{R})$ are the plant noise and measurement noise, respectively, with covariance $\mathcal{Q}\geq 0, \mathcal{R}\geq 0$. {$A,B,C,D$ are constant matrices of the appropriate dimension.}

We first analyze the control performance loss. With the above linear plant, \eqref{eq.L_infty_gain} becomes
\begin{equation*}
 \|(\Delta Y(t,x_{\bar t_0},   \xi )  )_\tau\|_{\mathcal{L}_2} \leq [\mathcal{L}^+_{\delta 2}]^\xi_{y'}\|(  \xi(t))_\tau\|_{\mathcal{L}_2},
\end{equation*} 
where $\Delta Y(t,x_{\bar t_0}, \xi )$ is the output trajectory of 
\begin{equation*} 
 \begin{split}
  \Delta\dot{   x'}&=(A-BK)\Delta  x'+BG \xi \\
  \Delta  y' &= (C-DK)\Delta x'+DG \xi.
\end{split} 
 \end{equation*}
This means that the incremental gain $[\mathcal{L}^+_{\delta 2}]^\xi_{y'}$ reduces to its non-incremental version, i.e., the $[\mathcal{L}^+_{2}]^\xi_{y'}$ gain.
In addition, we can obtain the similar results if we analyze the detection performance in linear case, that is,  the incremental gain $[\mathcal{L}^{-}_{\delta 2}]^{\xi}_{\hat y}$ reduces to  $[\mathcal{L}_{2}^{-}]^{  \xi}_{  \hat y}$.

Then for designing the watermark gain, we can use optimization problems \eqref{eq.opti_water2} and \eqref{eq.opti_main} with non-incremental gains $[\mathcal{L}^+_{2}]^\xi_{y}$ and $[\mathcal{L}_{2}^{-}]^{  \xi}_{  \hat y}$, and employ Algorithm 1 and Algorithm 2 to solve them, respectively, with $A_x=A$.

\begin{remark}
{ Most existing watermark-based detection methods, e.g., \cite{mo2009secure,mo2013detecting,zhai2021switching,FANG2020108698,zhao2023stochastic,chen2023replay},  follow a traditional framework that requires  explicit calculation of the  detection performance and control system performance loss. Typically, the   detection performance is calculated by the exact difference of innovation variances, e.g., \cite[Theorem 6]{mo2013detecting}, and the control system performance loss is calculated by the exact additional LQG cost, e.g., \cite[Theorem 5]{mo2013detecting}. Both calculations crucially rely on the linear    property of the Gaussian watermark. In contrast, our framework avoids this explicit calculation by   employing the lower bound of the detection performance and the upper bound of the control system performance loss. This difference allows for the generalization to the nonlinear scenarios.}

{
In the linear system scenario, our framework has  two main advantages compared with the       traditional  framework mentioned above. First, we do not limit the watermark to be i.i.d. Gaussian noise, allowing for a broader class of watermark signals to be employed.
 In addition, with the LMI technique, we are able to extend our framework to co-design the watermark with controller and observer, whereas this kind of extension is not easy for the literature mentioned above. 
However, we acknowledge that for linear systems with Gaussian noise, using the exact difference of innovation variances and the exact additional LQG cost can yield better performance than our framework. }

\end{remark}

  \begin{figure}[b]
\centering
 \includegraphics[width=1\columnwidth]{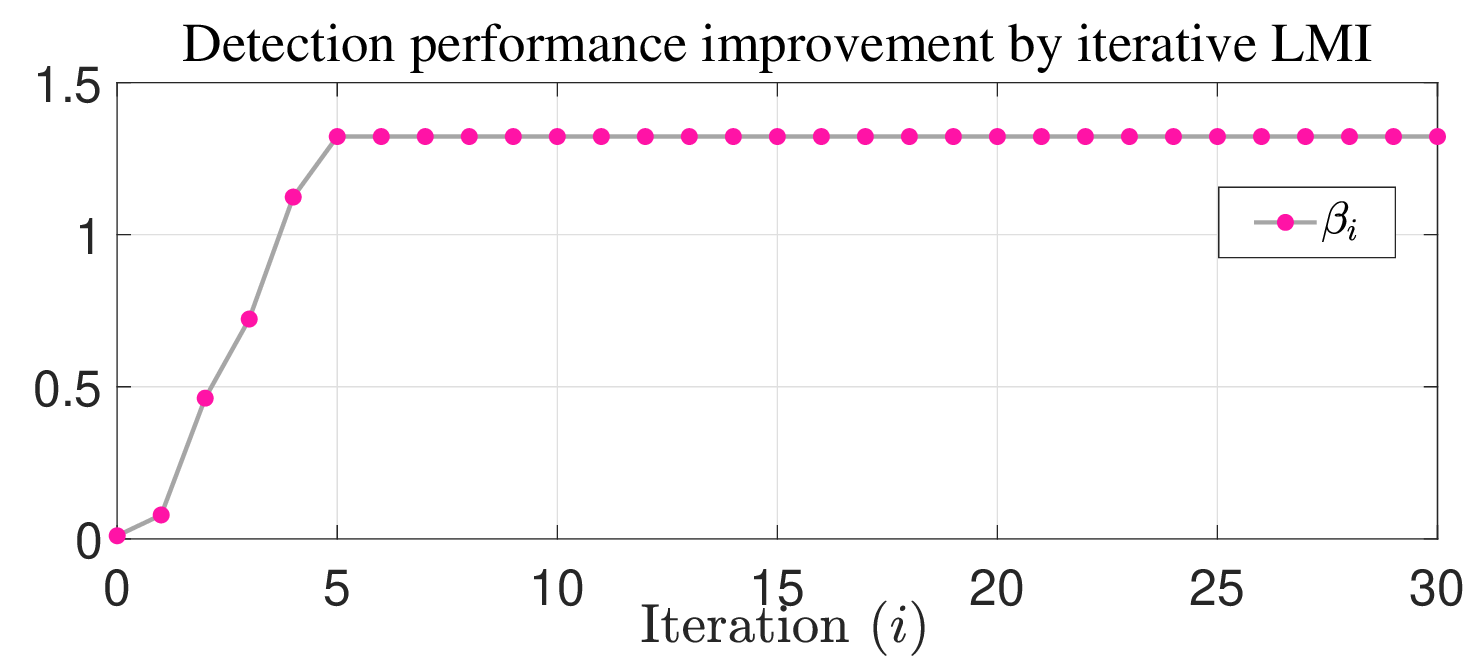}
\caption{The $\beta_i$ solved by iterative LMI.}
\label{fig.beta}
\end{figure}
\section{Numerical Simulation}\label{sec.simulation}
In this section, we validate the proposed framework by considering a single-link robot system \cite[Section 4.10]{isidori1985nonlinear} of the form \eqref{eq.plant}  with
\begin{equation*}
\begin{split}
      & f(x)=\begin{bmatrix}
        x_2\\
        -\frac{k}{J_2}x_1-\frac{mgd}{J_2}\text{cos}(x_1)-\frac{F_2}{J_2}x_2+\frac{k}{J_2 b}x_3\\
        x_4\\
        \frac{k}{J_1 b}x_1-\frac{k}{J_2 b^2}x_3-\frac{F_1}{J_1}x_4 
    \end{bmatrix}\\
     & B = \begin{bmatrix}
          0\! &0 \! &0\!  &\! 1
      \end{bmatrix}^\top, C =\begin{bmatrix}
          1 \! &0\!  &\! 0 &\! 0
      \end{bmatrix}, D = \begin{bmatrix}
         1\!  &0\!  &0 \! &0
      \end{bmatrix}^\top
  \end{split}  
\end{equation*}
where $k =0.4 kg\cdot m^2/s^2 $, $m=0.1kg$, $g=9.81m/s^2$, $d=0.1m$, $b = 2$,  $F_1 = 0.1 kg\cdot m^2/s $, $F_2=0.7kg\cdot m^2/s$, $J_1 = 0.15kg\cdot m^2$, $J_2 = 0.2kg\cdot m^2$ are parameters of the plant. In addition, the noise bounds $\overbar \nu=\overbar \omega=0.05$. { The observer and controller take forms of \eqref{eq.observer} and \eqref{eq.controller}, respectively, with $\kappa(\hat x) = -K\hat x$.}

In this simulation, we focus on the co-designing of the watermark, controller and observer and employing      Algorithm 2.
Letting $\alpha=4$, $G_\text{init}=1$ and according to Step 1 of Algorithm 2, the original $K$, $L$, and $G$ are calculated as
\begin{equation*}
\begin{split}
    &   K_0 = \begin{bmatrix}
        1.1525  &  0.1535 &   0.0755 &   0.1651
    \end{bmatrix}\\
   & L_0=\begin{bmatrix}
        0.3974  &  3.8453 &   0.2410&    0.9848
    \end{bmatrix} ^\top,~~~G_0 = 1.
\end{split} 
\end{equation*}
  
Then letting $N=30$ and choosing $\beta_0=0.01$,   $K_{\text{opt}}, L_{\text{opt}}$ and $G_{\text{opt}}$ solved by Algorithm 2 are
\begin{equation*}
\begin{split}
    &   K_{\text{opt}} = \begin{bmatrix}
        1.1426   &  0.4401  &   0.0673  &   0.1895
    \end{bmatrix}\\
   & L_{\text{opt}}=\begin{bmatrix}
        0.1623& 3.7539& 0.1596& 0.7642 
    \end{bmatrix} ^\top,~G_{\text{opt}} = 1.33.
\end{split} 
\end{equation*}

The $\beta_i$ obtained at each iteration of the algorithm is shown in Fig. \ref{fig.beta}, which represents the detection performance improvement by the iterative LMI.

In the following, both chaotic watermark and  Bernoulli watermark are tested. 
The chaotic watermark is  generated by a \textit{R$\ddot{o}$ssler prototype-4 system} \cite{sprott2010elegant}  of form \eqref{eq.chaotic_system} with
\begin{equation*} 
\begin{split}
&A=\begin{bmatrix}
     0&-1  & -1\\
1 &0 & 0 \\
0&0.5&-0.5
\end{bmatrix}, \phi(\theta)=\begin{bmatrix}
    0\\
0\\
-0.5\theta_2^2
\end{bmatrix}, \Lambda =\begin{bmatrix}
    0.5\\
    0\\
    0
\end{bmatrix}^\top.
 \end{split}
\end{equation*}

 \begin{figure}[!t]
\centering
 \includegraphics[width=1\columnwidth]{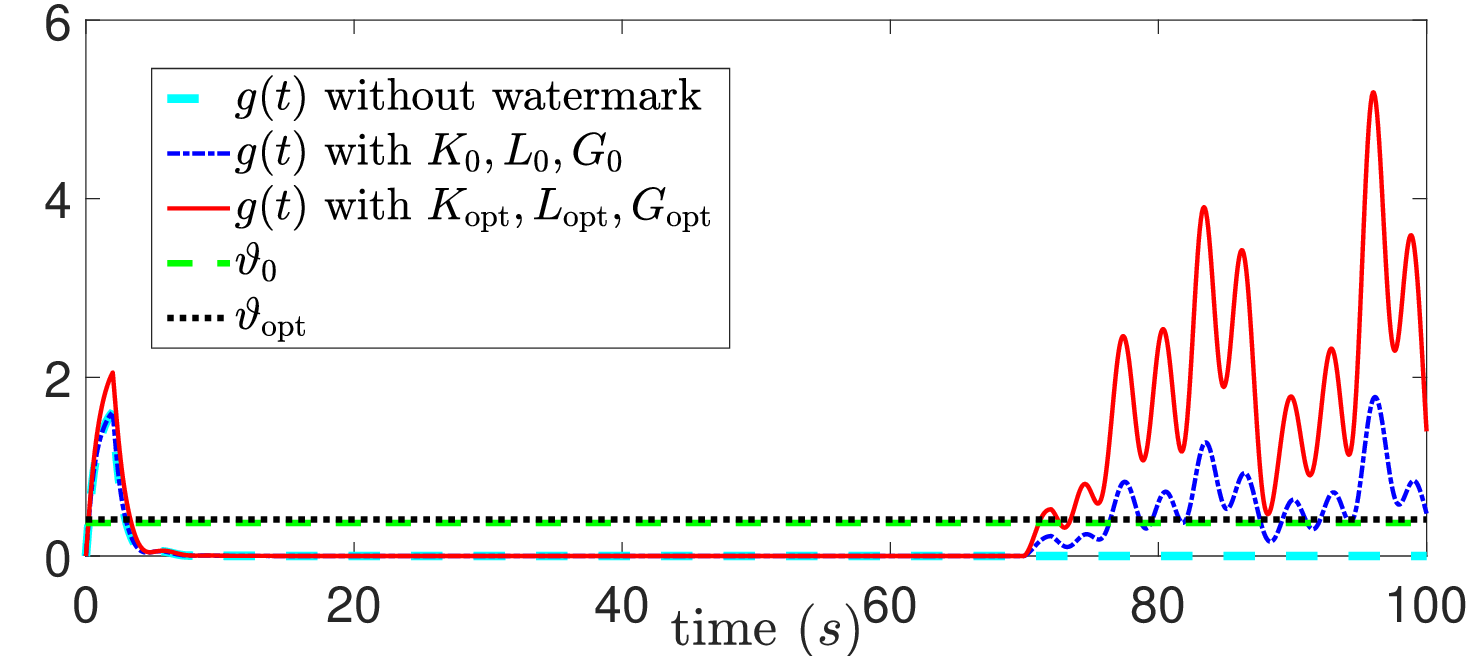}
 \caption{The detection performance with chaotic watermark.}
\label{fig.dete}
\end{figure}

\begin{figure}[!t]
\centering
\subfigure[]{\includegraphics[width=0.49\columnwidth]{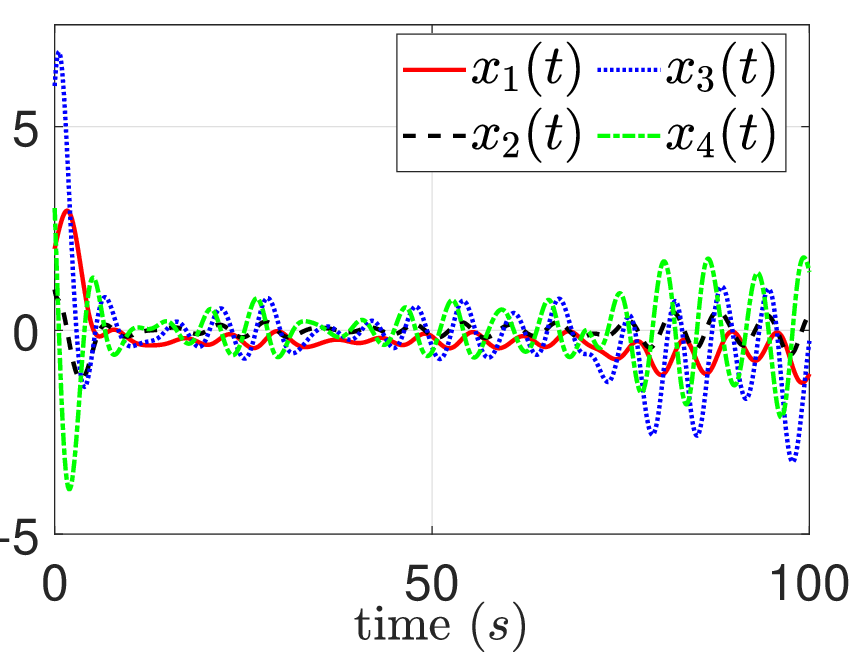}}
\subfigure[]{\includegraphics[width=0.49\columnwidth]{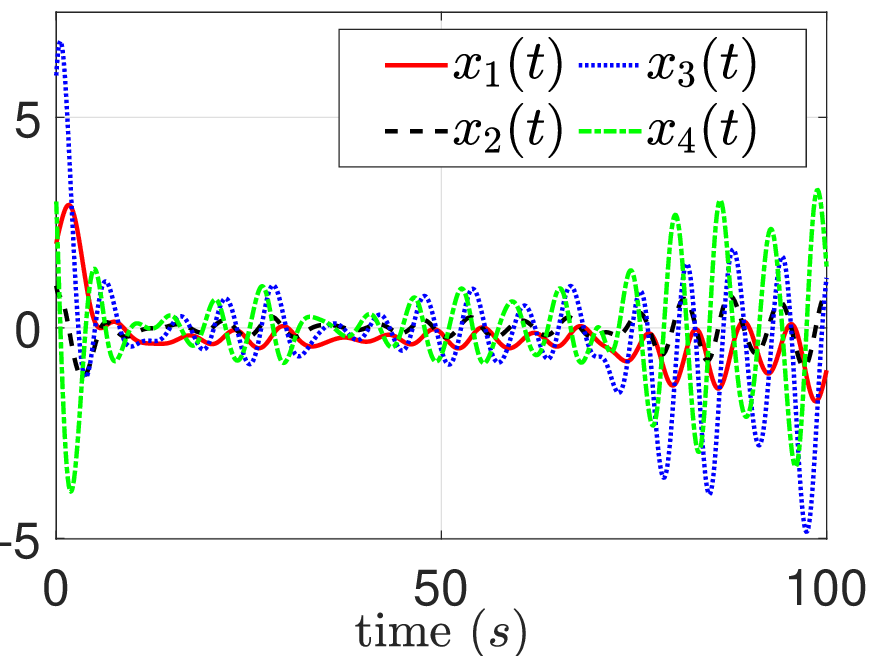}}
\caption{The states of the plant with chaotic watermark. (a) With   $K_0, L_0$ and $G_0$. (b) With  $K_\text{opt}, L_\text{opt}$ and $G_\text{opt}$.}
\label{fig.state}
\end{figure}

First of all, the detection performance of the innovation-based detector is illustrated in  Fig. \ref{fig.dete}, where $\sigma = 2$ and $\vartheta$ takes $ \|C\tilde x\|_\infty^2$ for simplicity, which is  calculated according to \cite[Theorem 4.19]{khalil2002nonlinear}. In Fig. \ref{fig.dete}, $\vartheta_{\text{opt}}$ and $\vartheta_{0}$ are the thresholds calculated by the optimized parameters $K_\text{opt}, L_\text{opt}$ and $G_\text{opt}$ and the original parameters $K_0$, $L_0$ and $G_0$, respectively.
It is shown that, without watermark, even under attack (from 70$s$), $g(t)$  is always smaller than the thresholds and thus the attack cannot be detected.

As illustrated in Fig. \ref{fig.dete}, with the chaotic watermark, $g(t)$ with both the optimized parameters and the original parameters exceed their respective thresholds, demonstrating an improvement in detection performance. Furthermore, $g(t)$ with the optimized parameters is significantly greater than that with the original parameters, while the performance of the control system using both sets of parameters is similar (see Fig. \ref{fig.state}). This verifies the effectiveness of Algorithm 2.

 \begin{figure}[!t]
\centering
 \includegraphics[width=1\columnwidth]{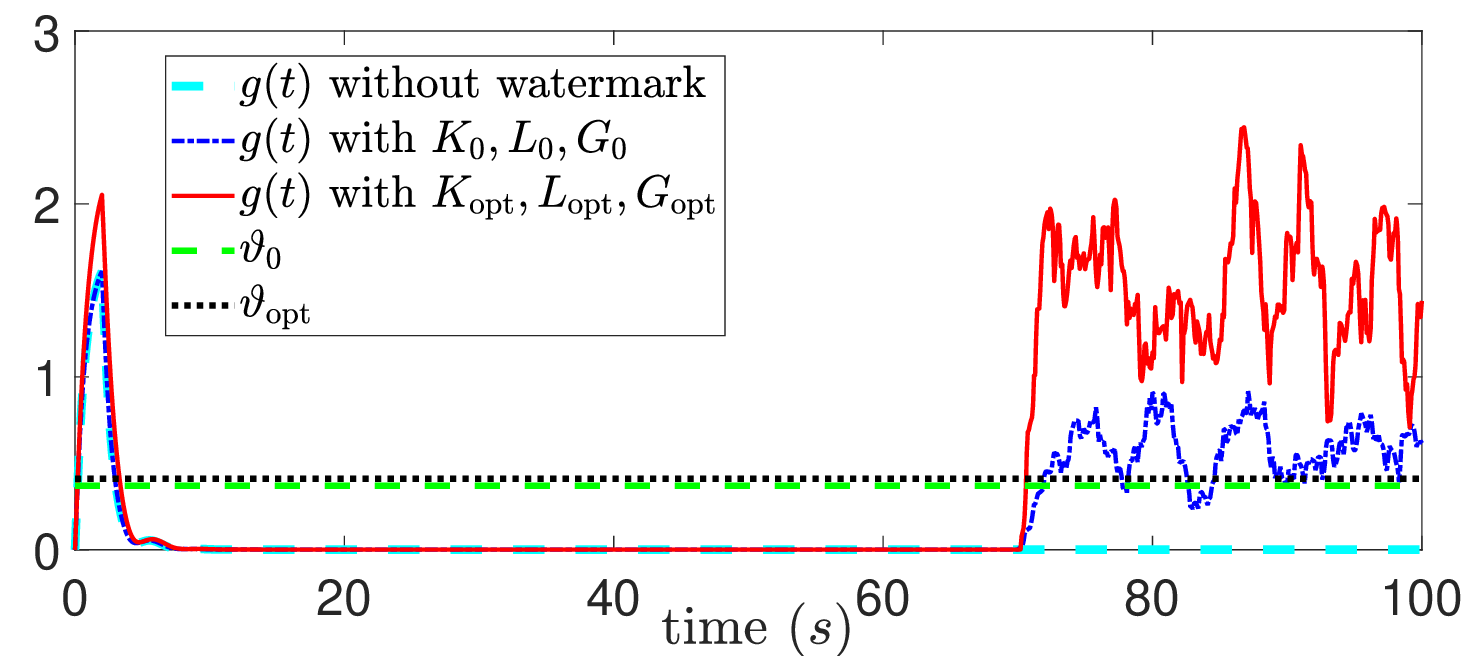}
 \caption{The detection performance with Bernoulli watermark.}
\label{fig.dete2}
\end{figure}

\begin{figure}[!t]
\centering
\subfigure[]{\includegraphics[width=0.49\columnwidth]{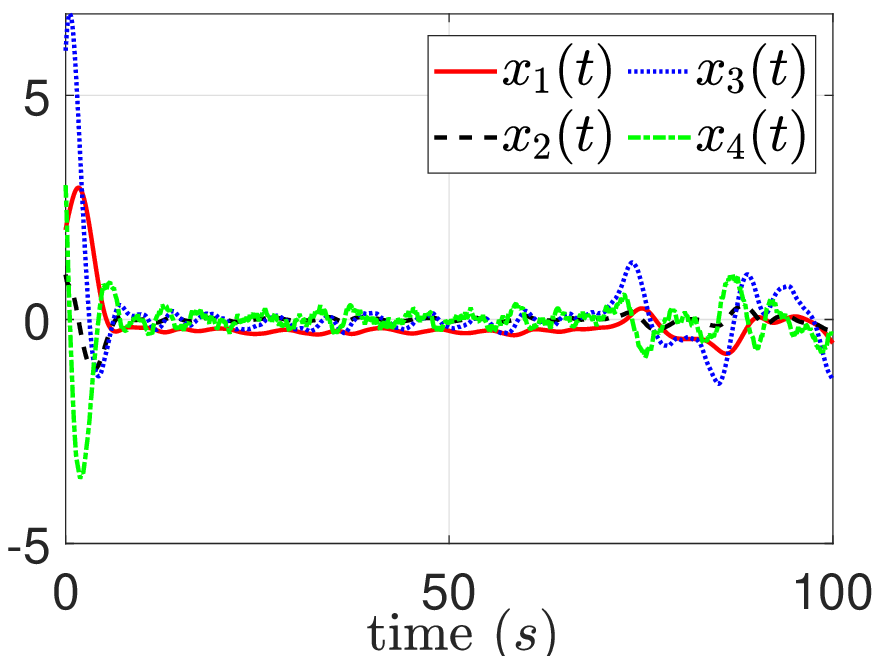}}
\subfigure[]{\includegraphics[width=0.49\columnwidth]{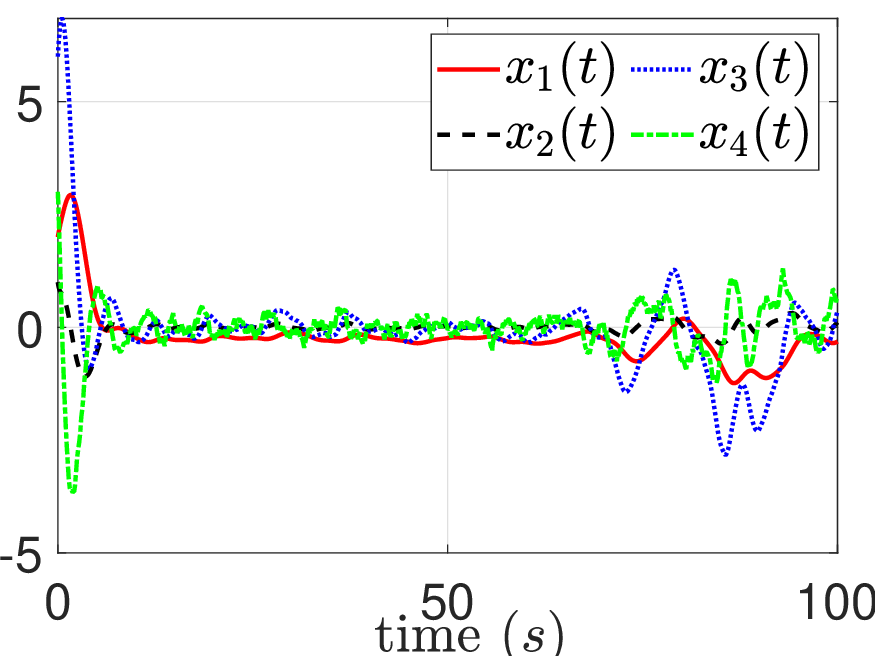}}
\caption{The states of the plant with Bernoulli watermark. (a) With   $K_0, L_0$ and $G_0$. (b) With  $K_\text{opt}, L_\text{opt}$ and $G_\text{opt}$.}
\label{fig.state2}
\end{figure}
Similar results can be obtained by analyzing the detection results with a Bernoulli watermark, as shown  in Figs. \ref{fig.dete2}-\ref{fig.state2}, where $\delta_t=0.1 s$. { Besides, we ran both the chaotic watermark and Bernoulli watermark  simulations 1000 times to get their statistical detection results, i.e., the detection rate, average detection delay (average d.d) and maximal detection delay (maximal d.d.), which are summarized in Tables I and II. In these simulations we can see that the proposed co-design methodology improves the detection rate and reduces the detection delay of the replay attack, without compromising the control performance bound established by $\alpha$. Furthermore, we can see that if we allow some additional performance loss (by increasing $\alpha$) we can additionally improve the detection indicators of the system.}

 \begin{table}[htbp]
	\centering  
{	\caption{Statistic Results for Chaotic Watermark}  
 \label{table1}  
	\scalebox{0.95}{\begin{tabular}{cc|ccc}  
		\toprule   
		& &\!Detection rate  \!& \!Average d.d.   & \!Maximal d.d \! \\  
		\hline
		 \multirow{2}{*}{$\alpha = 4$}& \!  $K_0$, $L_0$, $G_0$& 100\% &   3.79 s  \!&  15.50 s \! \\
	 &  $K_\text{opt}, L_\text{opt}$, $G_\text{opt}$ \!&100\% \!&   1.92s  &  10.61s    \\
     \hline
		 \multirow{2}{*}{$\alpha = 1$}&   $K_0$, $L_0$, $G_0$&87.4\%  &    10.28 s  &  29.78 s    \\
	 & $K_\text{opt}, L_\text{opt}$, $G_\text{opt}$ \!& 97.3\% \!&   8.34 s  & 29.74 s   \\
		\bottomrule 
	\end{tabular}}}
\end{table}

\begin{table}[htbp]
	\centering  
{		\caption{Statistic Results for Bernoulli Watermark}  
 \label{table2}  
	\scalebox{0.95 }{\begin{tabular}{cc|cccc}  
		\toprule   
		 &\!& \!Detection rate  \!& \!Average d.d. & \!Maximal d.d.  \\  
		\hline
	\multirow{2}{*}{$\alpha = 4$}&\!  $K_0$, $L_0$, $G_0$  & 100\%  &   1.22 s   &  4.37 s \! \\
    &   \!  $K_\text{opt}, L_\text{opt}$, $G_\text{opt}$ &100\%  &    0.66 s  &  1.64 s    \\
		\hline
		\multirow{2}{*}{$\alpha = 1$}&  $K_0$, $L_0$, $G_0$ \!&69.5 \%  &   13.15 s  &    29.87 s    \\ 
	&  $K_\text{opt}, L_\text{opt}$, $G_\text{opt}$&99.60\%  &   6.01 s   &   29.66 s  \\
		\bottomrule 
	\end{tabular}}}
\end{table}

In addition, to further show the improvement of detection performance in different noise levels, the following index is used, i.e.,
\begin{equation*}
\mathcal{D}:=\frac{\max_{t\in[70,100]}    g(t) }{\max_{t\in[20,69]}   g(t) } 
\end{equation*}
which is the ratio of detector's maximal output after and before the replay attack. A larger $\mathcal{D}$ implies a smaller $g(t)$ in the absence of the attack, and a larger $g(t)$ when the attack occurs, thereby indicating a better detection performance.
 \begin{figure}[!h]
\centering
\includegraphics[width=1\columnwidth]{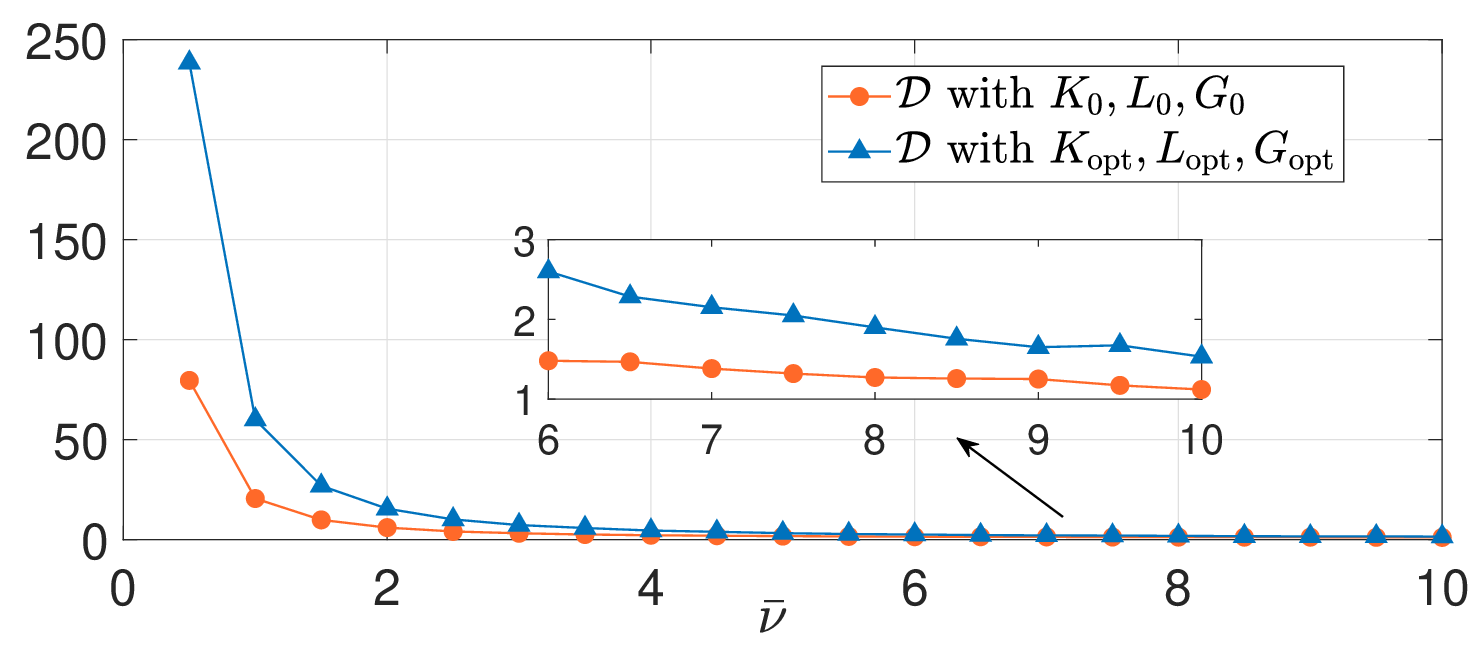}
\caption{The $\mathcal{D}$  with initial and optimized $K, L$ and $G$ with different noise level.}
\label{fig.D_noise}
\end{figure}

 The $\mathcal{D}$ obtained by the simulation in different noise levels is given in Fig. \ref{fig.D_noise} (with chaotic watermark). It is clear that $\mathcal{D}$ with optimized parameters is always greater than that with initial parameters, indicating that the detection performance is improved in all noise levels.

\section{Conclusion}\label{sec.conclution}
 This paper presented a novel watermark design framework for detecting replay attacks in nonlinear plants. First, we showed that the innovation-based detector may fail to detect replay attacks   and thus a watermark-based detection was motivated.  Then, to evaluate the detection performance and the control system performance loss induced by the watermark,  an incremental gain framework  were introduced. Using the incremental gains, we effectively evaluated both detection performance and control system performance loss. To balance these two factors, an optimization problem was formulated, which was then solved by establishing sufficient conditions for specific detection and performance loss objectives. Additionally, we extended the framework to co-design the watermark, controller, and observer. Finally, the proposed framework was validated by simulations.

\section*{References}
\vspace{-0.65cm}
\bibliographystyle{ieeetr}
\bibliography{bibio.bib}

\appendix 




\setcounter{equation}{0}

\renewcommand\theequation{A.\arabic{equation}} 
\subsection{Proof of Proposition \ref{def.Lyapunov_L_minus}}\label{prof.def.Lyapunov_L_minus}

  The  inequality  \eqref{eq.ineq_Vminus} is equivalent to
\begin{equation}\label{eq.dot_V-}
     \dot V^-( x_1,x_2,t) \leq -\gamma^-\|u_1-u_2\|^2 + \|y_1-y_2\|^2
  \end{equation}
  Taking integral for both sides of \eqref{eq.dot_V-} from any $t_0\geq0$ to $t>t_0$, we have 
\begin{equation}\label{eq.V-_int}\begin{split}
      & V^-( x_1,x_2,t) -V^-( x_1,x_2,t_0) \\
       &\leq \int_{t_0}^t  -\gamma^-\|u_1(s)-u_2(s)\|^2 + \|y_1(s)-y_2(s)\|^2 ds.
  \end{split}
  \end{equation}
Then with \eqref{eq.sandwich_Vminus},  $V^-( x_1,x_2,t)$ and $V^-( x_1,x_2,t_0)$ satisfy $V^-( x_1,x_2,t)\geq    \underline\alpha^-(\|x_1-x_2\|)\geq 0$ and $V^-( x_1,x_2,t_0)\leq  \overbar\alpha^-(\|x_1-x_2\|) $, respectively, and thus \eqref{eq.V-_int} is further deduced as
  \begin{equation}\label{eq.int_y}
    \begin{split}
     \int_{t_0}^t \|y_1(s)-y_2(s)\|^2 ds  &   \geq \int_{t_0}^t  \gamma^-\|u_1(s)-u_2(s)\|^2 ds\\     & - \overbar\alpha^-(\|x_1-x_2\|), \quad \forall t\geq t_0.
  \end{split}  
  \end{equation} 
The proof is complete by recalling the definition of $\|\cdot\|_{\mathcal{L}2}$  and incremental $\mathcal{L}^-_{\delta 2}$ gain.
 
 \setcounter{equation}{0}

\renewcommand\theequation{B.\arabic{equation}} 
\subsection{Proof of Lemma \ref{lem.L_minus_ori}}\label{prof.Hi-}

Define $\tilde x := x_1-x_2 $, $\tilde u := u_1-u_2$, $\tilde y :=y_1-y_2$.
Let $ V^-(x_1,x_2)=-\tilde x^{\rm T}Q\tilde x$. Since $Q$ is   symmetric negative, 
\begin{equation}
  -\lambda_{\max}(Q)\|\tilde x\|^2 \leq  V^-(x_1,x_2)\leq -\lambda_{\min}(Q)\|\tilde x\|^2
\end{equation}
which verifies condition \eqref{eq.sandwich_Vminus} for all $x_1,x_2\in \RR^n$.

We now proceed to verify   condition \eqref{eq.ineq_Vminus}.
  Taking the time derivative of $V^-(x_1,x_2) $, we have 
\begin{equation}\label{eq.dot_V_non}
\begin{split}
     \dot V^-(x_1,x_2)  &= -2 \tilde x^{\rm T}Q[f(x_1 )+Bu_1 -f(x_2)-Bu_2
      ].
\end{split}
\end{equation}

{According to the  Fundamental Theorem of Calculus   \cite[Chapter 6]{rudin1976principles}, we have 
\begin{equation}
    g(1)-g(0) = \int_0^1 \frac{\partial g}{\partial s}(s)ds
\end{equation}
for any $g \in\mathcal{C}^1$. Hence, letting $g(s): = f((s-1)\tilde x+x_1)$ and }recalling that   { $f$ $\in\mathcal{C}^1$}, we have that 
\begin{equation}\label{eq.f_x}
\begin{split}
    &f(x_1 ) - f(x_2 ) = \left( \int_0^1 \frac{\partial f}{\partial x}(\eta)ds \right)\tilde x  
\end{split}
\end{equation}
 where $\eta:=(s-1)\tilde x+x_1$. 
 Substituting \eqref{eq.f_x}  into  \eqref{eq.dot_V_non}, we have
\begin{equation}\label{eq.dot_V_non2}
\begin{split}
     &\dot V^-(x_1,x_2) \!\!=\!\!-2\tilde x^{\rm T}Q\left[ \left( \int_0^1 \frac{\partial f}{\partial x}(\eta)ds \right)\tilde x \!+\! B\tilde u \right]\\
     &=-\begin{bmatrix}
        \tilde x\\
        \tilde u 
    \end{bmatrix}^{\rm T} \int_0^1 \begin{bmatrix}
 Q\frac{\partial f}{\partial x}(\eta)+\frac{\partial f^{\rm T}}{\partial x}(\eta)Q &   QB \\
 *&0
    \end{bmatrix} ds \begin{bmatrix}
        \tilde x  \\
        \tilde u 
    \end{bmatrix}.
  \end{split}
\end{equation}
   
Notice that the inequality \eqref{eq.LMI_L_minus_ori} is equivalent to  
\begin{equation}\label{eq.LMI_non2}
 \begin{bmatrix}
        A^{\rm T}_x Q+QA_x&QB \\
        \star&0 \\
    \end{bmatrix}\succeq   \begin{bmatrix}
         -C^{\rm T}  C & -C^{\rm T} D \\
       \star&    -D^{\rm T}     D +\gamma^- I\\
    \end{bmatrix}  .
    \end{equation}
 Combining this bound with \eqref{eq.dot_V_non2}, we obtain
\begin{equation}\label{eq.dot_V_non3}
\begin{split}
     &\dot V^-   \!\! \leq\!\!-\begin{bmatrix}
        \tilde x\\
        \tilde u 
    \end{bmatrix}^{\rm T}\!\!\int_0^1\begin{bmatrix}
-C^{\rm T} C &-C^{\rm T}D \\
\star&-D^{\rm T} D  + \gamma^- I
    \end{bmatrix} \!\!ds\!\!\begin{bmatrix}
        \tilde x  \\
        \tilde u 
    \end{bmatrix}  \\
  &  = -\begin{bmatrix}
        \tilde x\\
        \tilde u 
    \end{bmatrix}^{\rm T}\begin{bmatrix}
        *\end{bmatrix}^{\rm T}\begin{bmatrix}
\gamma^- I&0\\
0&-I
    \end{bmatrix} \begin{bmatrix}
        0&I\\
         C&D
    \end{bmatrix}\begin{bmatrix}
        \tilde x  \\
        \tilde u 
    \end{bmatrix}\int_0^1 ds. \\
\end{split}
\end{equation}
   With the fact that
    $y_1-y_2= C\tilde x+D\tilde u$, \eqref{eq.dot_V_non3} is equivalent to
\begin{equation}\label{eq.dot_V_non4}
\begin{split}
     &\dot V^-(x_1,x_2)   \leq -\begin{bmatrix}
        \tilde u ^{\rm T}& \tilde y^{\rm T}
    \end{bmatrix}  \begin{bmatrix}
\gamma^- I&0\\
0&-I
    \end{bmatrix}  \begin{bmatrix}
        \tilde u  \\
        \tilde y 
    \end{bmatrix}   \\
    &=-\gamma^- \|u_1-u_2\|^2+\|y_1-y_2\|^2
\end{split}
\end{equation}
which verifies the   condition \eqref{eq.ineq_Vminus}. 
Now both \eqref{eq.sandwich_Vminus} and \eqref{eq.ineq_Vminus}   have been verified, which, according to Proposition \ref{def.Lyapunov_L_minus}, verifies that
 system \eqref{eq.general_sys} 
    has an incremental 
    $\mathcal{L}^-_{\delta 2}$ gain $ [\mathcal{L}^-_{\delta 2}]^u_y\geq  \gamma^-$.

 \setcounter{equation}{0}

\renewcommand\theequation{C.\arabic{equation}} 
\subsection{Proof of Proposition \ref{pro.replay}}\label{prof.fail_detection}
With controller \eqref{eq.controller}, the healthy observer \eqref{eq.observer} evolves
\begin{equation}\label{eq.observer_with_water2_v_0}
\begin{split}
      \dot{\hat { x}} &= f(\hat x)+B\kappa(\hat x    )   -L( C\hat x +D\kappa(\hat x ))  +Ly  \\
 \hat y  &= C\hat x +D\kappa(\hat x ),
 \end{split}
\end{equation}
whose solution in $t \in [0, T)$ can be denoted as   $\widehat X(t-T,\hat x_0,v , y)$ with  $t \in [T, 2T)$.

By comparing \eqref{eq.observer_with_water_fir} and \eqref{eq.observer_with_water2_v_0} and with the $\delta$ISS property and $v=0$ in mind, there exists a $\mathcal{K}\mathcal{L}$ function $ \check \beta$ such that, for $t\in[T,2T)$,
\begin{equation}\label{eq.hat_bX_minus}
\begin{split}
    \|\widehat X(t,\hat x_T,0,y^a)-\widehat X(t-T,\hat x_0,0,y )\|&\leq   \check \beta( \|  \hat x(T )-\hat{x}(0 )\|,t).
\end{split}
     \end{equation} 
    In addition, the innovation difference between \eqref{eq.observer_with_water_fir} and \eqref{eq.observer_with_water2_v_0} satisfies
\begin{equation} \label{eq.ra_r'}
\begin{split}
      & \|\widetilde Y(t,\overbar \bx_0,0, y^a) -\widetilde Y(t-T,\bx_0,0, y)  \|\\
      &=   \| - \widehat Y(t,\hat  x_T,0, y^a )+\widehat Y(t-T,\hat  x_0, 0,y )\|\\
     &=  \|- C\widehat X(t,\hat x_T,0, y^a)-D\kappa(\widehat X(t,\hat x_T,0, y^a)) \\
     &+C\widehat X(t-T,\hat x_0,0, y)+D\kappa(\widehat X(t-T,\hat x_0,0, y))\|.
\end{split}
     \end{equation} 

 Recalling the Lipschitz  property     \eqref{eq.lip_k}
 and combining   \eqref{eq.hat_bX_minus} and \eqref{eq.ra_r'},  we have
\begin{equation} \label{eq.ra_r''}
\begin{split}
      \|\widetilde Y(t,\overbar \bx_0,0, y^a) & -\widetilde Y(t-T,\bx_0,0, y)  \|\\
      &\leq (\|C\|+\|D\|l_\kappa)\check{\beta} ( \|  \hat x(T )-\hat{x}(0 )\|,t).
      \end{split}
     \end{equation} 
We first analyze the output of the detector when $t\in[ T +\sigma,2T)$. In this case, the output difference between the attacked detector  and the history healthy one is  
\begin{equation}\label{eq.g_g'1}
\begin{split}
      & \frac{1}{\sigma} \!\!\int_{t-\sigma}^t   \!\! \|\widetilde Y(s,  \overbar\bx_0, v, y^a )\|^2 d s   \!-  \!\frac{1}{\sigma} \int_{t-\sigma}^t  \!\!   \|\widetilde Y(s-T,   \bx_0, v, y )\|^2 d s \\
        &= \frac{1}{\sigma} \int_{t-\sigma}^t  (\|\widetilde Y(s,   \overbar \bx_0, 0, y^a )\|-    \|\widetilde Y(s-T,   \bx_0, 0, y )\| ) \\
        &\qquad(\|\widetilde Y(s,   \overbar \bx_0, 0, y^a )\|+       \|\widetilde Y(s-T,   \bx_0, 0, y )\| ) d s.
\end{split}
\end{equation}
 From \eqref{eq.ra_r''}, we obtain  that $\widetilde Y(s,   \overbar \bx_0, 0, y^a )$ is   bounded, as $\widetilde Y(t-T,\bx_0,0, y)$ is the  healthy  output, which is bounded according to Assumption \ref{ass.delta_ISS}. Consequently, there exists a constant $ M>0$ such that $(\|\widetilde Y(s,   \overbar \bx_0, 0, y^a )\|+       \|\widetilde Y(s-T,   \bx_0, 0, y )\| ) \leq M,\forall s\in[T,2T)$. In addition, with the triangle inequality $\|\widetilde Y(s,   \overbar \bx_0, 0, y^a )\|-    \|\widetilde Y(s-T,   \bx_0, 0, y )\| \leq  \|\widetilde Y(s,\overbar \bx_0,0, y^a) -\widetilde Y(s-T,\bx_0,0, y)  \|$ and \eqref{eq.ra_r''} in mind, \eqref{eq.g_g'1} can be further deduced as 
\begin{equation}\label{eq.output_dete_0}
\begin{split}
      & \frac{1}{\sigma}  \!\!\int_{t-\sigma}^t  \!\! \|\widetilde Y(s,   \overbar \bx_0, v, y^a )\|^2 d s -\frac{1}{\sigma} \!\! \int_{t-\sigma}^t    \|\widetilde Y(s-T,   \bx_0, v, y )\|^2 d s\\
       &\leq M  \frac{1}{\sigma} \int_{s-\sigma}^t   \|\widetilde Y(s,\overbar \bx_0,0, y^a) -\widetilde Y(s-T,\bx_0,0, y^a)  \|   d s \\
       &\leq M     \frac{1}{\sigma} \int_{t-\sigma}^t  (\|C\|+\|D\|l_\kappa)\check{\beta}( \|\hat x(T)-\hat{x}(0)\|,s)d s  \\
       &\leq M       (\|C\|+\|D\|l_\kappa)\check{\beta}( \|\hat x(T)-\hat{x}(0)\|,t-\sigma) 
\end{split}
\end{equation}
which implies
\begin{equation}\label{eq.output_dete_2}
\begin{split}
      &    g(t)\leq  \frac{1}{\sigma} \int_{t-\sigma}^t    \|\widetilde Y(s-T,   \bx_0, v, y )\|^2 d s +\Delta_1 
\end{split}
\end{equation}
where $\Delta_1 : =M      (\|C\|+\|D\|l_\kappa)\check{\beta}( \|\hat x(T)-\hat{x}(0)\|,t-\sigma)] $.

Then, for $T\leq t<T+\sigma $, the monitoring signal is
\begin{equation}\label{eq.g_t_T}
\begin{split}
      & g(t) =\frac{1}{\sigma}  \!\! \int^{T}_{t-\sigma}  \!\!  \|\widetilde Y(s,   \bx_0, v, y )\|^2 d s +\frac{1}{\sigma}  \!\! \int^{t}_T     \!\!\|\widetilde Y(s,   \overbar \bx_0, v, y^a )\|^2 d s  
\end{split}
\end{equation}
where the term ${1}/{\sigma}  \int_{t}^T    \|\widetilde Y(s,   \overbar \bx_0, v, y^a )\|^2 d s$ satisfies the following inequality by the similar deduction from \eqref{eq.g_g'1} to \eqref{eq.output_dete_0}, i.e.,
\begin{equation}\label{eq.ra_r'_T}
\begin{split}
      &   \frac{1}{\sigma} \int_{T}^t   \|\widetilde Y(s,   \overbar \bx_0, v, y^a )\|^2 d s -\frac{1}{\sigma} \int_{T}^t    \|\widetilde Y(s-T,   \bx_0, v, y )\|^2 d s \\
       &\qquad\leq M (\|C\|+\|D\|l_\kappa)\check{\beta}( \|\hat x(T)-\hat{x}'(T)\|,T).
\end{split}
\end{equation}

Combining \eqref{eq.g_t_T} and  \eqref{eq.ra_r'_T}, the monitoring signal satisfies 
\begin{equation}\label{eq.g_t_T_2}
\begin{split}
       g(t) \leq &   \frac{1}{\sigma}  \int^{T}_{t-\sigma}   \|\widetilde Y(s,   \bx_0, v, y )\|^2 d s\\
      &+  \frac{1}{\sigma} \int_{T}^t    \|\widetilde Y(s-T,   \bx_0, v, y )\|^2 d s +\Delta_2
\end{split}
\end{equation}
where $\Delta_2 := M   (\|C\|+\|D\| l_\kappa)\check{\beta}( \|\hat x(T)-\hat{x}'(T)\|,T)$. 

Finally, combining \eqref{eq.output_dete_2} for $t\in[ T+\sigma,2T)$ and \eqref{eq.g_t_T_2} for $t\in[T,T+\sigma)$,  \eqref{eq.g_t_dISS} is obtained with $\overbar\beta(\cdot,t):= M    (\|C\|+\|D\| l_\kappa)\check{\beta}(\cdot,t-\sigma)$.

\setcounter{equation}{0}

\renewcommand\theequation{D.\arabic{equation}} 
\subsection{Proof of Theorem \ref{theorem.output_difference}}\label{Appen.output_difference}

The proof is divided into two parts that prove \eqref{eq.dete_no_attack_pro} and \eqref{eq.dete_with_attack_pro}, respectively.

\noindent
\emph{Proof of \eqref{eq.dete_no_attack_pro}:}  

\noindent

According to Assumption \ref{ass.steady_state}   and recalling that $\|\omega\|_\infty\leq  \overbar \omega $ and $\|\nu \|_\infty\leq  \overbar \nu $, we have 
\begin{equation}
    \|\widetilde X(t,\bx_0,v, y)\|\leq\alpha_\omega (\overbar \omega)+\alpha_\nu(\overbar \nu),\quad \forall  t\in[0,\infty).
\end{equation}

Then   we have
\begin{equation}
\begin{split}
      \|\widetilde Y(s,   \bx_0, v, y )\|& =\|CX(t,x_0,u)+Du+\nu\\
      &-C\widehat X(t,\hat x_0,v , y )+Du\|\\
      &\leq \|C\|\|\widetilde X(t,\bx_0,v, y)\|+\overbar \nu\\
      &\leq  \|C\| (\alpha_\omega (\overbar \omega)+\alpha_\nu(\overbar \nu))+\overbar \nu.
\end{split}    
\end{equation}
As a result, \begin{equation*}
\begin{split}
   g(t) & = \frac{1}{\sigma}   \int_{t-\sigma}^t   \|\widetilde Y(s,   \bx_0, v, y )\|^2 d s\\
    &\leq [\|C\| (\alpha_\omega (\overbar \omega)+\alpha_\nu(\overbar \nu))+\overbar \nu]^2  
\end{split}
\end{equation*}
which proves \eqref{eq.dete_no_attack_pro}.

\noindent
\emph{Proof of \eqref{eq.dete_with_attack_pro}:}   

\noindent

With $v=G\xi$, the healthy observer \eqref{eq.observer}   becomes  
\begin{equation}\label{eq.observer_with_water2}
\begin{split}
      \dot{\hat { x}} &= f(\hat x)+B(\kappa(\hat x )+G\xi)  -L (C\hat x+D(\kappa(\hat x )+G\xi) +Ly \\
 \hat y  &= C\hat x +D(\kappa(\hat x )+G\xi), 
\end{split} 
\end{equation}
whose solution in $t \in [0, T)$ can be denoted as   $\widehat X(t-T,\hat x_0,v , y)$ with  $t \in [T, 2T)$.

 Comparing \eqref{eq.observer_with_water} and \eqref{eq.observer_with_water2}, we have
\begin{equation*} 
\begin{split}
    & \widetilde Y(t,\overbar \bx_{t_0},v, y^a) -\widetilde Y(t-T,\bx_{t_0},v, y)  \\ 
      &=     - \widehat Y(t,\hat  x_{t_0+T}, v,y^a )+\widehat Y(t-T,\hat  x_{t_0},v, y ).
\end{split}
\end{equation*}

Then, since the system \eqref{eq.observer_with_water} { has an $[\mathcal{L}^-_{\delta 2}]^\xi_{\hat y}$ gain}, for $0\leq t_0\leq  T$,
\begin{equation}\label{eq.tilde_Y_L_minu}
\begin{split}
   &   {\|(\widetilde Y(t,\overbar \bx_{t_0},v, y^a) -\widetilde Y(t-T,\bx_{t_0},v, y ) )_\tau\|_{\mathcal{L}_2}} \\
   &= {\|( \widehat Y(t,\hat  x_{t_0+T},v, y^a )-\widehat Y(t-T,\hat  x_{t_0},v, y ))_\tau\|_{\mathcal{L}_2}}\\
      &\geq ([\mathcal{L}^-_{\delta 2}]^\xi_{\hat y}-\varepsilon)  {\|(\xi(t)-\xi(t-T))_\tau\|_{\mathcal{L}_2}}\\
      &\qquad\qquad-\overbar \alpha^-_a(\|\hat x(t_0+T)-\hat x(t_0 )\|) ,
\end{split}
\end{equation}
for all initial states $\overbar \bx_{t_0}, \bx_{t_0}$ and inputs $\xi(t), \xi(t-T)$. 

We first consider the case $ t \in[  T+\sigma,2T) $. In this case, let $t_0$ and $\tau$ in \eqref{eq.tilde_Y_L_minu} be $t_0= t-T-\sigma$ and  $\tau  = \sigma$. Then we have 
\begin{equation}\label{eq.dete_t_geq_T_sigma}
\begin{split}
\int^{ t}_{t-\sigma} & \!\!\!\!  \|\widetilde Y(t,\overbar \bx_{t_0},\!v, y^a)\|^2 ds  \!\! \geq \!\!-\!\!   \int^{t }_{t-\sigma} \!\!\!\!  \|\widetilde Y(t\!-\!T,\bx_{t_0},\!v, y)\|^2 d s\\
&+([\mathcal{L}^-_{\delta 2}]^\xi_{\hat y}-\varepsilon) \int^{t}_{t-\sigma}      \|\xi(s)-\xi(s\! -\! T)\|^2 d s  \!\\
&- \overbar \alpha^-_a(\|\hat x(t-\sigma)-\hat x(t-\sigma-T)\|),
\end{split}
\end{equation}
where the inequality $\int^{ t}_{t-\sigma}    \|\widetilde Y(t,\overbar \bx_{t_0},v, y^a)\|^2 d s+ \int^{ t}_{t-\sigma}  \|\widetilde Y(t-T,\bx_{t_0},v, y )\|^2 d s\geq  \int^{ t}_{t-\sigma}  \| \widetilde Y(t,\overbar \bx_{t_0},v, y^a)-\widetilde Y(t-T,\bx_{t_0},v, y )\|^2 d s$ is used. Then combining \eqref{eq.dete_t_geq_T_sigma} and \eqref{eq.dete_no_attack_pro}, i.e., $ 1/\sigma \int^{T+\sigma}_{T}    \|\widetilde Y(t-T,\bx_0,v, y)\|^2 d s\leq g_n$, and with the fact $\int^{ t}_{t-\sigma}     \|\widetilde Y(t,\overbar \bx_{t_0},v, y^a)\|^2 ds=\int^{ t}_{t-\sigma}     \|\widetilde Y(t,\overbar \bx_{0},v, y^a)\|^2 ds, \forall t\in[ T+\sigma,2T)$ in mind, we have, for   $t\in[ T+\sigma,2T)$,
\begin{equation}\label{eq.dete_t_geq_T_sigma2}
\begin{split}
 g(t)&\geq    \frac{[\mathcal{L}^-_{\delta 2}]^\xi_{\hat y}-\varepsilon }{\sigma}  \int^{ t}_{t-\sigma}   \|\xi(s)-\xi(s-T)\|^2 d s\\
  &-\frac{\overbar \alpha^-_a(\|\hat x(t-\sigma)-\hat x(t-\sigma-T)\|)}{\sigma} - g_n.
\end{split}
\end{equation}

Now we consider the case $t\in[T, T+\sigma)$. In this case, 
\begin{equation}\label{eq.g_t_leq}
\begin{split}
    g(t)&\!\!=\!\!\frac{1}{\sigma}\!\!\int^{ T}_{t-\sigma}\!\!\|\widetilde Y(t, \bx_{0},v, y)\|^2 ds  \!\!+\frac{1}{\sigma}\int^{ t}_{T}   \!\!\|\widetilde Y(t,\overbar \bx_{0},v, y^a)\|^2 ds. 
\end{split}
\end{equation}
Let $t_0$ and $\tau$ in \eqref{eq.tilde_Y_L_minu} be $t_0= 0$ and  $\tau  = \sigma$. Similar to \eqref{eq.dete_t_geq_T_sigma}, we have 
\begin{equation}\label{eq.tildeY_T_t}
\begin{split}
\int^{ t}_{T} & \!\! \! \|\widetilde Y(t,\overbar \bx_{0},v, y^a)\|^2 ds \!\!  \geq -\!\!   \int^{t }_{T} \!\! \!   \|\widetilde Y(t-T,\bx_{0},v, y)\|^2 d s\\
&+([\mathcal{L}^-_{\delta 2}]^\xi_{\hat y}-\varepsilon) \int^{t}_{T}      \|\xi(s)-\xi(s\! -\! T)\|^2 d s  \!\\
&- \overbar \alpha^-_a(\|\hat x(T)-\hat x(0)\|).
\end{split}
\end{equation}

Substituting \eqref{eq.tildeY_T_t} into \eqref{eq.g_t_leq}, we have
\begin{equation}\label{eq.g_t_leq2}
\begin{split}
 &   g(t) \geq  \frac{1}{\sigma}\!\!\int^{ T}_{t-\sigma}\!\!\|\widetilde Y(t, \bx_{0},v, y)\|^2 ds  \!-\!\frac{1}{\sigma} \!\!\int^{t }_{T} \!\!  \|\widetilde Y(t-T,\bx_{0},v, y)\|^2 d s\\
&+\frac{([\mathcal{L}^-_{\delta 2}]^\xi_{\hat y}\!-\!\varepsilon)}{\sigma} \int^{t}_{T}   \! \!  \|\xi(s)-\xi(s\! -\! T)\|^2 d s \! -\! \frac{\overbar \alpha^-_a(\|\hat x(T)-\hat x(0)\|)}{\sigma}.\\
\end{split}
\end{equation}

By 
$\frac{1}{\sigma}\int^{ T}_{t-\sigma}\|\widetilde Y(t, \bx_{0},v, y)\|^2 ds \geq 0$ and $  \frac{1}{\sigma} \int^{t }_{T}    \|\widetilde Y(t-T,\bx_{0},v, y)\|^2 d s \leq g_n$, signal \eqref{eq.g_t_leq2},    for $t\in[T, T+\sigma)$, satisfies
    \begin{equation}\label{eq.g_t_leq3}
\begin{split}
    g(t)& \geq \frac{[\mathcal{L}^-_{\delta 2}]^\xi_{\hat y}-\varepsilon}{\sigma} \int^{t}_{T}      \|\xi(s)-\xi(s\! -\! T)\|^2 d s  \!\\
&- \frac{\overbar \alpha^-_a(\|\hat x(T)-\hat x(0)\|)}{\sigma}-g_n.\\
\end{split}
\end{equation} 
 
Finally, \eqref{eq.dete_with_attack_pro} is obtained by combining \eqref{eq.dete_t_geq_T_sigma2} and \eqref{eq.g_t_leq3}.

\renewcommand\theequation{E.\arabic{equation}} 
\subsection{ Proof of Proposition \ref{propo.L_infinity2} } \label{proof.L_infinity}

By Lemma \ref{lem.L_inf_ori} and the Schur complement, $[\mathcal{L}^+_{\delta 2}]^\xi_{y'}(G)  \leq  \alpha$ if there exists a   symmetric positive definite matrix $P_s$ and matrix $G$ such that, for all $x\in \mathbb{R}^n $,
    \begin{equation}\label{eq.LMI_L_infty_2}
        \begin{bmatrix}
      \mathcal{N}_{11}'  &P  B  G&(C -D  K)^\top\\
        \star&-\alpha I_m &(D  G)^\top  \\
        \star&\star&- I_n
    \end{bmatrix}\preceq 0
    \end{equation}
    where 
$\mathcal{N}_{11}': = A_x^\top P +P  A_x  -  (B  K)^\top P- PB  K+\epsilon P$ with $\epsilon>0$.
Then  by pre and post multiplying $\text{diag}(P^{-1},I_m,I_n)$ on \eqref{eq.LMI_L_infty_2} and letting $P_s= P^{-1}$, we obtain \eqref{eq.LMI_L_infinityG}.

\renewcommand\theequation{F.\arabic{equation}} 
\subsection{Proof of Proposition  \ref{pro.G_minus}} \label{appendix.prof_G_minus}
(\eqref{eq.LMI_L_minus}$\Rightarrow $\eqref{eq.LMI_L_minusG}): 
By the Schur complement lemma, LMI \eqref{eq.LMI_L_minus} is equivalent to
 \begin{equation*} 
\Pi_2:=  \begin{bmatrix}
            \bm\bar\mathcal{M}_{11}&  ( C - D   K)^\top D   G&-Q \\
      \star& \bm\bar\mathcal{M}_{22} &(B  G-LD  G)^\top\\
      \star&\star&I
       \end{bmatrix} \succeq 0 
       \end{equation*}
      where $ \bm\bar\mathcal{M}_{11}= QQ+  Q  A_x+  A_x^\top   Q- Q    B   K   - (B  K)^\top  Q- QLC  - (LC )^\top Q +QLD  K+(LD  K)^\top Q + ( C - D  K)^\top  (C - D   K) $ and $\bm\bar\mathcal{M}_{22}= (B  G-LD  G)^\top(B  G-LD  G) +G^\top  D  ^\top   D   G-\beta I_m$.
      
    Then, \eqref{eq.LMI_L_minus} implies \eqref{eq.LMI_L_minusG} if choosing $Q_{0}=Q$ and $G_0=G$.

(\eqref{eq.LMI_L_minus}$\Leftarrow $\eqref{eq.LMI_L_minusG}): Inequality \eqref{eq.LMI_L_minusG} can be rewritten as 
\begin{equation*}
    \Pi_2-\text{diag}[\Lambda_1, \Lambda_2,0]\succeq 0
\end{equation*} 
where $\Lambda_1:=(Q-Q_0)(Q-Q_0)$, $\Lambda_2:=(G-G_0)^\top[(B  -LD   )^\top(B  -LD   )+ D  ^\top   D  ](G-G_0)$,
which implies $\Pi_2\succeq 0$.

\renewcommand\theequation{G.\arabic{equation}} 
\subsection{ Proof of Lemma \ref{lem.contraction_1} } \label{proof.contra}

Let $\tilde x: = x -\hat x $. We first verify that  $ \tilde {x}=0$ is exponentially stable  when $\omega=\nu =0$.
Define  $V :=  \tilde {x}^\top R \tilde {x} $, then
\begin{equation*}
 \begin{split}
      \dot V  &= 2 \tilde x^\top R[f(x)+Bu-f(\hat x)-Bu  +L(Cx-C\hat x)].
 \end{split} 
 \end{equation*} 
In addition, \eqref{eq.obsv2} is equivalent to
\begin{equation*}
   (  A_x - LC   )^\top  R+R      ( A_x - LC) \preceq -\epsilon_1 R, 
\end{equation*}
 which implies \cite[Lemma 3.1]{FB-c}
\begin{equation*}
\begin{split}
    &\tilde x^\top R[f(x)-f(\hat x)  +L(Cx-C\hat x)]\leq -\frac{\epsilon_1}{2} \tilde x^\top R\tilde x\\
\end{split}
\end{equation*}
 Exponential stability follows from the fact that
 \begin{equation}\label{eq.observer_contract}
     \dot V\leq - \epsilon_1  V.
 \end{equation}
For the case $\omega\neq 0$ and $\nu\neq 0$, the ISS  property of the error system \eqref{eq.error_observer} follows from \cite[Corollary 3.17]{FB-c}.

Now we verify that the nonlinear plant \eqref{eq.plant} is $\delta$ISS w.r.t. inputs $\tilde x, v,\omega $.  
Let  $\overbar V :=   (x_1-x_2)^\top S  (x_1-x_2)$.  
  Since  { $f\in   \mathcal{C}^1$},  similar to \eqref{eq.f_x}, we have
\begin{equation*}
\begin{aligned}
    &f(x_1)+BK  x_1   - (f(x_2)+BK  x_2)\\
    &= \left( \int_0^1 \frac{\partial f}{\partial x}(\eta) ds \right)(x_1-x_2) - BK(  x_1-  x_2).
\end{aligned}
\end{equation*}

\color{black}

With \eqref{eq.ctrl2} in mind,  the time derivative to $\overbar V$ is
\begin{equation*}
 \begin{split}
      \dot {\overbar V} &\leq  2(x_1-x_2)^\top S  (f(x_1)+BK  x_1   - (f(x_2)+BK  x_2))    \\
&\leq -\epsilon_2 {\overbar V}, 
 \end{split}
\end{equation*}
which proves that $x_1$ exponentially converges to $x_2$  \cite[Theorem 4.10]{khalil2002nonlinear}. Finally, we can obtain that  the nonlinear plant \eqref{eq.plant} is $\delta$ISS w.r.t. inputs $\tilde x,v,\omega $ from \cite[Corollary 3.17]{FB-c}.

\color{black}

\begin{IEEEbiography}[{\includegraphics[width=1in,height=1.25in,clip,keepaspectratio]{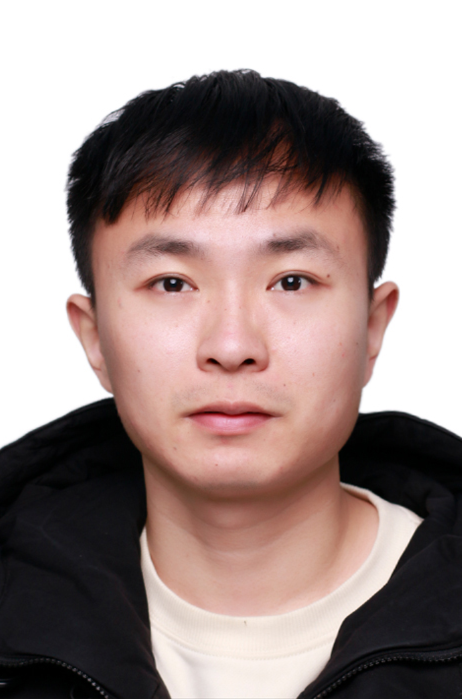}}]{Tao Chen}  received the B.Eng. degree in automation from Shandong Agricultural University, China, in 2018, and M.S. degree in Control Science and Engineering from South China University of Technology, China, in 2021. He is currently pursuing the Ph.D. degree in Control Science and Engineering in Zhejiang University, China. 
   His current research interests include cyber-physical systems, secure estimation, and attack detection.
\end{IEEEbiography}

\begin{IEEEbiography}[{\includegraphics[width=1in,height=1.25in,clip,keepaspectratio]{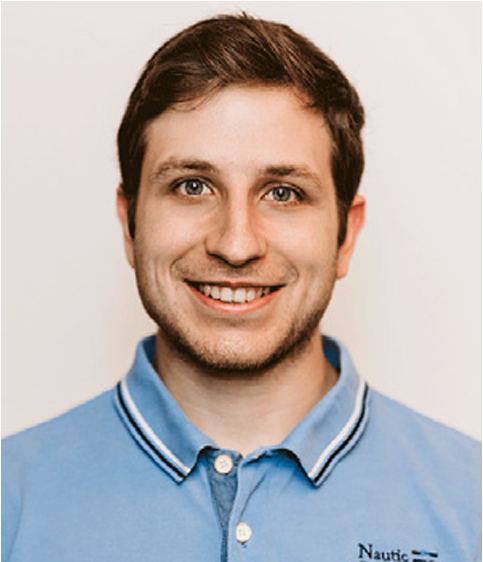}}]{Andreu Cecilia}   received the B.Eng. degree in industrial engineering, the double M.Sc. degree in automatic control/industrial engineering and the Ph.D. in automatic control from the Universitat Polit\`{e}cnica de Catalunya, Barcelona, Spain, in 2017, 2020 and 2022, respectively. In 2022-2023, he worked as a post-doctoral researcher at LAGEPP, Lyon, France. He is currently working as a lecturer at Universitat Polit\`{e}cnica de Catalunya, Barcelona. His research interests include observers, nonlinear system theory and its application to energy systems and cyber-security.
\end{IEEEbiography}

\begin{IEEEbiography}[{\includegraphics[width=1in,height=1.25in,clip,keepaspectratio]{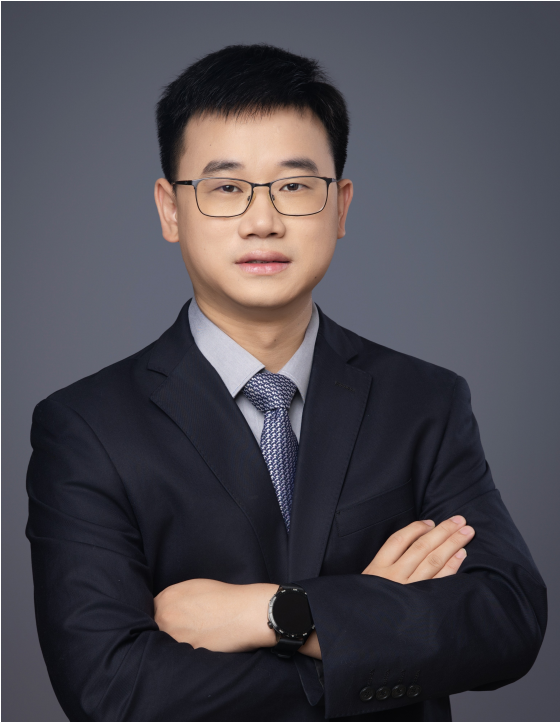}}]{Lei Wang } received the B.Eng. degree in automation from Wuhan University, China, in 2011, and Ph.D. degree in Control Science and Engineering from Zhejiang University, China in 2016. From December 2014 to December 2015, he visited C.A.SY.-DEIS, University of Bologna as a visiting Ph.D. student. 
    Lei held research positions with School of Electrical and Electronic Engineering at Nanyang Technological University, Singapore, School of Electrical Engineering and Computing at University of Newcastle, Australia, and Australian Center for Field Robotics, The University of Sydney, Australia. Since November 2021 he has been a Hundred-Talent Researcher at College of Control Science and Engineering, Zhejiang University, China. Lei serves as an AE of Journal of Control and Decision, and a member of IFAC Technical Committee 2.3 Nonlinear Control Systems, and has served as an IPC member of several conferences. His current research interest lies in the development of nonlinear control theory from nonlinear systems to networked systems, with applications to fuel-cell systems and power systems.
\end{IEEEbiography}
 \begin{IEEEbiography}[{\includegraphics[width=1in,height=1.25in,clip,keepaspectratio]{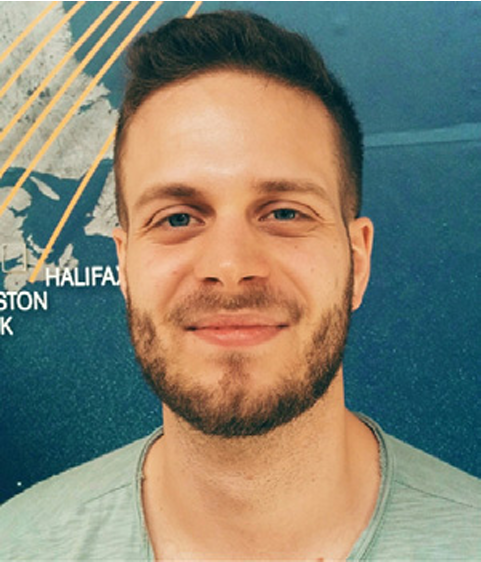}}] {Daniele Astolfi}  received the B.S. and M.S. degrees in automation engineering from the University of Bologna,
 Italy, in 2009 and 2012, respectively. He obtained a joint Ph.D. degree in Control Theory from the University of Bologna, Italy, and from Mines ParisTech, France, in 2016. In 2016 and 2017, he has been a Research Assistant at the University of Lorraine (CRAN), Nancy, France. Since 2018, he is a CNRS Researcher at LAGEPP, Lyon, France. His research interests include observer design, feedback stabilization and output regulation for nonlinear systems, networked control systems, hybrid systems, and multi-agent systems. He serves as an associate editor of the IFAC journal Automatica. He was a recipient of the 2016 Best Italian Ph.D. Thesis Award in Automatica given by Societ\`{a} Italiana Docenti e Ricercatori in Automatica (SIDRA, Italian Society of Professors and Researchers in Automation Engineering).
\end{IEEEbiography}

 \begin{IEEEbiography}[{\includegraphics[width=1in,height=1.25in,clip,keepaspectratio]{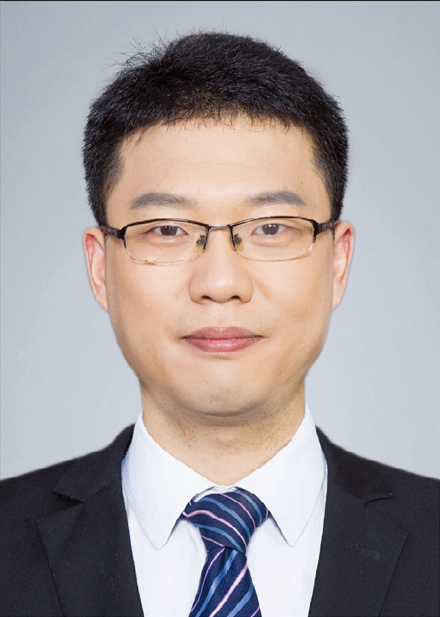}}]{Zhitao Liu }    received the B.S. degree from Shandong University at Weihai, China, in 2005, and the Ph.D. degree in control science and engineering from Zhejiang University, Hangzhou, China, in 2010.	
    From 2011 to 2014, he was a Research Fellow with TUM CREATE, Singapore. He was an Assistant Professor from 2015 to 2016 and an Associate Professor from 2017 to 2021 in Zhejiang University, where he is currently a Professor with the Institute of Cyber-Systems and Control, Zhejiang University. His current research interests include robust adaptive control, wireless transfer systems, and energy management systems.

\end{IEEEbiography}

\end{document}